\documentclass[11pt]{article}
\usepackage{amsmath,amssymb,amsthm}
\usepackage{graphicx}
\usepackage[numbers]{natbib}
\usepackage[margin=1in, top=0.2in, headheight=10pt, headsep=5pt, includehead]{geometry}
\usepackage{xcolor}
\usepackage[colorlinks=true, citecolor=blue, linkcolor=red]{hyperref}
\usepackage{titlesec}

\definecolor{tocgreen}{RGB}{34, 139, 34}

\makeatletter
\AtBeginDocument{%
  \let\orig@starttoc\@starttoc
  \renewcommand{\@starttoc}[1]{%
    \begingroup
    \hypersetup{linkcolor=tocgreen}%
    \orig@starttoc{#1}%
    \endgroup
  }%
}
\makeatother
\usepackage[T1]{fontenc}
\usepackage{newtxtext, newtxmath}  
\usepackage{tocloft}              

\numberwithin{equation}{section}
\titleformat{\section}
  {\normalfont\Large\bfseries\MakeUppercase}{\thesection}{1em}{| }

\titleformat{\subsection}
  {\normalfont\large\bfseries}{\thesubsection}{1em}{| }

\titleformat{name=\section,numberless}
  {\normalfont\Large\bfseries}{}{0pt}{}

\makeatletter
\newtheoremstyle{dotstyle}{3pt}{3pt}{\normalfont}{}{\bfseries}{.}{ }{\thmname{#1}\thmnumber{ #2}\thmnote{ (#3)}}
\theoremstyle{dotstyle}
\newtheorem{theorem}{Theorem}[section]
\newtheorem{lemma}{Lemma}[section]
\newtheorem{proposition}{Proposition}[section]
\newtheorem{definition}{Definition}[section]
\newtheorem{remark}{Remark}[section]
\newtheorem{corollary}{Corollary}[section]
 
\renewenvironment{proof}[1][\proofname]{%
  \par\pushQED{\qed}%
  \topsep6\p@\@plus6\p@\relax
  \trivlist
  \item[\hskip\labelsep\normalfont\bfseries #1\@addpunct{.}]%
  \ignorespaces
}{%
  \popQED\endtrivlist\@endpefalse
}
\makeatother

\usepackage{etoolbox}   
\appto\appendix{
  \titleformat{\section}                         
    {\normalfont\Large\bfseries\MakeUppercase}  
    {APPENDIX~\thesection:}{1em}{}              
  \titleformat{name=\section,numberless}         
    {\normalfont\Large\bfseries}{}{0pt}{}
}

\title{\bfseries A Two-Parameter Memory-Weighted Velocity Operator for Time and State Variables: Foundations and Fundamental Properties}
\author{
  Jiahao Jiang\thanks{Correspondence to: Jiahao Jiang, School of Mathematics, Southwest Jiaotong University. Email: \href{mailto:jiahao.jiang@swjtu.edu.cn}{jiahao.jiang@swjtu.edu.cn}} \\
  \textit{School of Mathematics, Southwest Jiaotong University, Chengdu 610031, China}
}
\date{}

\begin{document}

\maketitle

\begin{center}
\large\bfseries Abstract
\end{center}
\vspace{10pt}
We introduce and analyze a memory-weighted velocity operator \(\mathscr{V}_{\alpha,\beta}\) within the framework of operator theory, establishing rigorous theoretical results for systems with time-varying memory. The operator employs two independent continuous exponents \(\alpha(t)\) and \(\beta(t)\) that separately weight past state increments and elapsed time scaling. This decoupling mechanism is the main novelty of the work: it allows the two memory aspects to evolve independently, which is essential for systems where the influence of past states and the perception of time may change in qualitatively different ways. Such situations arise naturally in adaptive materials, non-stationary transport, and evolving memory processes.

Motivated by systems across multiple physical contexts—such as viscoelastic materials with stress-dependent relaxation or anomalous transport with history-dependent characteristics—the framework addresses memory aspects that may evolve differently across disciplines.

We establish the operator's foundational properties: an explicit integral representation, linearity, and continuous dependence on the memory exponents with respect to uniform convergence. Central to the analysis are weighted pointwise estimates revealing how the exponent difference \(\beta(t)-\alpha(t)\) modulates the operator, and we establish weighted boundedness estimates for this linear operator. These estimates exhibit a natural compensation mechanism between the two memory weightings.

For the uniform-memory case \(\alpha=\beta\equiv1\), we prove that \(\mathscr{V}_{\alpha,\beta}[x](t)\) asymptotically recovers the classical derivative \(\dot{x}(0)\) as \(t\to 0^{+}\), ensuring consistency with local calculus. The presentation is self-contained and includes detailed proofs of all results, with technical appendices providing auxiliary tools of independent interest. By decoupling the memory weighting of state increments from that of elapsed time, the proposed operator offers a new perspective on memory-dependent systems, and may serve as a modest addition to the analytical methods available in operator theory and its applications.

\vspace{10pt}
\noindent \textbf{Keywords:} memory-dependent systems, time-varying orders, weighted estimates, continuous dependence, operator theory

\vspace{10pt}  
{\centering\Large\bfseries\MakeUppercase{Contents}\par}
\vspace{6pt}
\makeatletter
\@starttoc{toc}
\makeatother
\vspace{12pt}  

\section{INTRODUCTION}

Fractional calculus has emerged as a powerful framework for describing systems with memory and nonlocal effects, where the order of differentiation or integration can be non-integer. Within this broad context, the present work introduces a memory-weighted velocity operator that incorporates time-varying, power-law memory into the notion of a rate of change. A key feature of the proposed operator is the use of two independent, time-varying exponents that separately weight past state increments and elapsed time, thereby allowing the two aspects of memory to evolve independently. This decoupling mechanism is motivated by systems where the influence of history and the scaling of time may change in qualitatively different ways, a situation that falls naturally within the scope of fractional calculus and its interdisciplinary applications. Before presenting the operator and its analysis, we briefly recall the classical derivative and several directions in which generalized velocity concepts have been developed.

The measurement of change is a cornerstone of dynamical systems theory. The classical derivative, defined as a local limit of difference quotients, provides an instantaneous and memoryless rate of change. While supremely effective for Markovian processes and systems governed by local laws, many phenomena across physics, biology, and engineering exhibit a dependence of the present state on its history \cite{bhat1992characterization}. This necessitates the development of generalized velocity concepts that can encapsulate such non-local, memory-dependent dynamics. Investigations into such generalized operators span diverse fields, from the algebraic characterization of velocity maps in operator algebras \cite{bhat1992characterization} and their role in quantum mechanics on noncommutative spaces \cite{kovavcik2013velocity}, to the analysis of global existence for transport equations driven by nonlocal velocity fields \cite{bae2015global}. The challenge of even computing fundamental differential operators (divergence, curl) from the motion of discrete particle clouds further underscores the intricate relationship between observation, averaging, and the definition of velocity in collective systems \cite{maurel2024computing}. A related development in operator theory has introduced delay-coordinate maps and kernel integral operator techniques to construct coherent observables and approximate eigenfunctions of evolution operators, providing a framework for spectral analysis in ergodic dynamical systems \cite{giannakis2021delay}.
Complementing this direction, operatorial methods based on creation and annihilation operators have been applied to interacting populations and migration processes, where algebraic structures provide a systematic framework for handling constraints such as population upper bounds \cite{bagarello2013phenomenological}. Along a similar line, Toeplitz operators on the space of real analytic functions have been studied, including injectivity properties in the spirit of the Coburn–Simonenko theorem and invertibility characterizations via Fredholm index zero \cite{jasiczak2018coburn}. A further development concerns the iterates of composition operators defined by polynomials acting on Gelfand–Shilov classes of ultradifferentiable functions, including topologizability properties and a complete characterization for affine operators \cite{albanese2025topologizability}.

To mathematically capture memory effects, fractional calculus has emerged as a principal framework. The analysis of evolution equations involving fractional Laplacians has led to deep insights into the regularity and long-time behavior of solutions \cite{vazquez2017classical, chen2010heat}. A significant advancement within this paradigm is the theory of **variable-order (VO) differential operators**, where the order of differentiation—and thus the intrinsic memory kernel—is allowed to vary in time or space. This concept, introduced for mechanical systems to model transitions between dynamic regimes \cite{coimbra2003mechanics}, has been rigorously developed for pseudo-differential equations \cite{umarov2009variable}. The physical interpretation of the variable order as a normalized phase shift in oscillatory systems connects the mathematical formalism directly to measurable dynamical properties \cite{ramirez2010selection}. Recent definitions of fractional derivatives with “variable memory” explicitly aim to bridge models with purely local (short) and full (global) historical dependence \cite{baliarsingh2022fractional}.

The modeling power of this VO framework is evidenced by its successful application to complex, multi-scale problems. It describes the propagation of electromagnetic waves in layered media with crossover properties \cite{kachhia2021electromagnetic}, provides accurate numerical solutions for nonlinear reaction-diffusion and advection-diffusion equations with variable parameters \cite{kumar2019gegenbauer}, and forms the basis for novel, hybrid deterministic-stochastic models of biological processes like diabetes progression \cite{hasan2025novel}. The statistical physics of systems possessing such time-varying memory characteristics has been formally articulated, introducing forces linked to gradients in the memory intensity itself \cite{kobelev2003statistical}. Analytically, the study of these systems often employs concepts like subordination, which decomposes complex memory-dependent kinetics into simpler processes, offering both physical insight and technical simplification \cite{iomin2025subordination, gorska2023subordination}.

Naturally, this leads to the study of discrete counterparts \cite{tarasov2024discrete}, numerical methods with memory for solving nonlinear equations \cite{petkovic2011family}, systems with finite, fixed memory lengths \cite{ledesma2023differential}, theoretical limits of finite-dimensional models for memory processes \cite{fanizza2024quantum}, and their manifestations in continuum mechanics, such as history-dependent hemivariational inequalities for non-stationary Navier-Stokes flows with long memory \cite{oultou2025history} and hyperbolic phase-field models incorporating memory kernels \cite{rotstein2001phase}. A substantial body of literature is dedicated to establishing the fundamental analytical properties of such memory systems. This includes research on stabilization and decay for memory-type thermoelasticity \cite{messaoudi2012general}, well-posedness and Mittag-Leffler stability for fractional heat conduction with fading memory \cite{kerbal2024well}, functional viability theorems for differential inclusions constrained by past states \cite{haddad1984functional}, and the convergence, boundedness, and ergodicity of regime-switching stochastic processes with infinite memory \cite{li2021convergence}. Further theoretical depth is provided by explorations of fractional Lagrangian mechanics \cite{jumarie2009lagrangian}, detailed asymptotic profiles for nonlocal heat equations with memory \cite{cortazar2024asymptotic}, and the development of novel analytical concepts like fractional convexity \cite{del2022fractional} alongside sharp regularity results (e.g., Hölder continuity) for solutions to distributed-order evolution equations \cite{kubica2024holder}.

A crucial aspect for both the rigorous analysis and practical application of memory models involves establishing foundational properties of solutions. Work in this area includes investigating the Cauchy problem for a fractional diffusion equation with a memory term, where the time-fractional derivative is taken in the Caputo sense, giving a representation of solutions under Fourier series, analyzing initial value problems, and discussing the stability of the fractional derivative order under some assumptions on the input data \cite{akdemir2023dependence}. Beyond such specific analyses, the study of memory effects facilitated by fractional operators permeates diverse domains. Comprehensive reviews have summarized solution methods and the memory effect in fractional-order chaotic systems \cite{he2022solutions}. In complex fluids, generalized hydrodynamic correlations and their associated fractional memory functions have been analyzed \cite{rodriguez2015generalized}. Comparisons have been made between the behaviors of systems described by fractional differential equations and those described by fractional difference equations, using examples like the Caputo standard \(\alpha\)-family of maps to investigate similarities and differences between power-law and factorial-law memory \cite{edelman2014caputo}. Furthermore, memory represented via fractional-order dynamics has been shown to alter properties such as alternans and spontaneous activity in minimal models of cardiac cells \cite{comlekoglu2017memory}. The theoretical exploration extends naturally to discrete-time evolution equations with memory \cite{ponce2021discrete}, control and stabilization problems in fractional evolution systems \cite{ammari2021stabilization}, and even to sophisticated settings combining memory, mean-field interactions, and singular stochastic control \cite{agram2019singular}. Concurrently, efforts to generalize the very foundations of fractional calculus continue, aiming to provide a more encompassing framework for anomalous stochastic processes \cite{jiang2025study}. The analytical study of singular operators has been further enriched by investigations into Sobolev spaces associated with a singular perturbation of the Laplace operator \cite{georgiev2025sobolev}, regularization and discretization of transfer operators via entropic optimal transport \cite{junge2024entropic}, connections between localization operators and Toeplitz operators \cite{englivs2009toeplitz}, and the classification of operators within small operator ideals \cite{manoussakis2021small}—each offering complementary perspectives that inform the analysis of memory-weighted operators such as the operator \(\mathscr{V}_{\alpha,\beta}\) introduced in Definition~\ref{def:velocity_operator}.

Synthesizing these threads—the need for flexible memory modeling via variable-order operators, the importance of establishing core analytical properties like continuous dependence, and the value of a self-contained mathematical construct—motivates the present work. Herein, we introduce and systematically analyze a **memory-weighted velocity operator** \(\mathscr{V}_{\alpha,\beta}\). Rather than embedding memory within a pre-specified differential equation, we propose a direct, operator-theoretic generalization of the derivative. Its definition centers on a time-varying power-law memory kernel, governed by two independent, continuous exponent functions \(\alpha(t)\) and \(\beta(t) \in (0,1]\). This design allows separate weighting of historical state differences (\(\alpha\)) and elapsed time (\(\beta\)), providing a versatile primitive for describing systems whose memory characteristics evolve dynamically.

The subsequent analysis is structured as follows. Section~\ref{sec:foundations_and_properties} lays the mathematical groundwork: we define the relevant function spaces, the time-varying memory kernel, and formally construct the operator \(\mathscr{V}_{\alpha,\beta}\). We then establish its elementary properties, including an explicit simplified form and linearity. A central theoretical result, Theorem~\ref{thm:continuous_dependence} in Subsection~\ref{subsec:continuous_dependence}, proves that for a fixed function \(x \in C^1\), the mapping \((\alpha, \beta) \mapsto \mathscr{V}_{\alpha,\beta}[x]\) is continuous with respect to uniform convergence of the exponents. This ensures that the velocity depends stably on the memory law, a crucial property for both mathematical analysis and physical modeling where memory parameters may be subject to uncertainty or gradual variation.

Subsection~\ref{subsec:weighted_estimates_boundedness} examines the mapping properties of \(\mathscr{V}_{\alpha,\beta}\) through the development of weighted pointwise estimates. Theorem~\ref{thm:weighted_boundedness} establishes how the exponent difference \(\beta(t)-\alpha(t)\) modulates the magnitude of \(\mathscr{V}_{\alpha,\beta}[x](t)\), revealing a natural compensation mechanism between the two memory aspects. These estimates lead to conditions under which \(\mathscr{V}_{\alpha,\beta}\) defines a bounded linear operator from \(\mathcal{C}^{1}(I)\) into \(\mathcal{C}(I)\), providing insight into the operator's functional-analytic behavior that is relevant for the study of evolution equations governed by such memory-weighted rates of change.

Subsection~\ref{subsec:uniform_memory_case} examines the important uniform-memory case \(\alpha=\beta\equiv 1\). We prove that in this regime, the operator asymptotically recovers the classical derivative at the initial time, providing a clear bridge between the non-local formulation and standard calculus. To ensure a self-contained and rigorous presentation, we have developed and compiled necessary technical tools in the appendices: these include foundational inequalities for controlling logarithmic singularities coupled with power functions (Appendix~\ref{app:log_power_inequalities}), a key lemma establishing the uniform convergence of products of function sequences (Appendix~\ref{app:uniform_convergence}), and a systematic, constructive framework for error-control functions tailored to Taylor remainders (Appendix~\ref{app:error_control}).

By introducing the memory-weighted velocity operator \(\mathscr{V}_{\alpha,\beta}\) and establishing its core mathematical properties---such as linearity, continuous dependence on the memory exponents, weighted estimates characterizing its boundedness, and asymptotic recovery of the classical derivative in the uniform-memory limit---this work provides a well-defined mathematical object for incorporating time-varying, power-law memory into the concept of a rate of change. The analytical framework developed here, together with the supporting technical lemmas, offers a self-contained foundation for further investigation. Potential future directions could include the formulation and analysis of differential equations whose dynamics are governed by this operator, which may find relevance in theoretical studies of processes with non-stationary or adaptively weighted memory, such as those encountered in certain viscoelastic or transport phenomena where distinct aspects of memory may evolve independently. In particular, the operator's linearity and stability properties make it a natural candidate for incorporation into nonlinear evolution equations, where memory effects may interact with state-dependent feedback or external forcing.

\section{MATHEMATICAL FOUNDATIONS AND BASIC PROPERTIES OF THE MEMORY-WEIGHTED VELOCITY OPERATOR} \label{sec:foundations_and_properties}
\subsection{Function Spaces and Notational Conventions} \label{sec:preliminaries}
Throughout this work we fix a finite time horizon $T > 0$ and denote 
$I := [0, T]$. We shall work primarily with the following classical spaces 
of real-valued functions on $I$:

\begin{itemize}
    \item $\mathcal{C}(I)$: the Banach space of continuous functions on $I$, 
          equipped with the supremum norm
          \begin{equation}\label{eq:C_norm}
          \|x\|_{\infty} := \sup_{t \in I} |x(t)|;
          \end{equation}
    
    \item $\mathcal{C}^{1}(I)$: the subspace of continuously differentiable functions; 
          that is, functions $x \in \mathcal{C}(I)$ for which the derivative 
          $\dot{x}(t)$ exists at each $t \in I$ (with the appropriate one-sided 
          derivatives understood at the endpoints) and satisfies $\dot{x} \in \mathcal{C}(I)$.
          We endow this space with the natural norm
          \begin{equation}\label{eq:C1_norm}
          \|x\|_{\mathcal{C}^{1}} := \|x\|_{\infty} + \|\dot{x}\|_{\infty},
          \end{equation}
          where $\|\dot{x}\|_{\infty}$ is defined analogously to \eqref{eq:C_norm}:
          \begin{equation}\label{eq:C1_derivative_norm}
          \|\dot{x}\|_{\infty} := \sup_{t \in I} |\dot{x}(t)|;
          \end{equation}
\end{itemize}

The symbol $\mathbb{R}$ denotes the field of real numbers.  
For the Gamma function we employ the standard definition
\begin{equation}\label{eq:Gamma_def}
\Gamma(z) := \int_{0}^{\infty} s^{z-1} e^{-s}\, ds \qquad (\Re(z) > 0),
\end{equation}
together with the functional equation
\begin{equation}\label{eq:Gamma_functional_eq}
\Gamma(z+1) = z\,\Gamma(z),
\end{equation}
which will be used frequently in what follows.

\subsection{Time-Varying Power-Law Memory Weight Functions} \label{sec:time_varying_weights}
The core of our construction is a family of weight functions that quantify how strongly 
the state of a system at a given instant depends on its past.  In the classical fractional 
calculus framework, the weight function decays according to a fixed power-law.  
Motivated by the desire to incorporate time-dependent memory characteristics, 
we consider a variant where the exponent of the power-law is permitted to vary 
with the current time.  This leads to the following definitions.

\begin{definition}[Time-varying memory exponent] \label{def:time_varying_exponent}
    A function $\varrho \in \mathcal{C}(I)$ with $\varrho(t) \in (0,1]$ for all 
    $t \in I$ is called a \textbf{time-varying memory exponent}.  
    The value $\varrho(t)$ measures the \textit{memory intensity} at time 
    $t$: values close to $0$ indicate a strong emphasis on the recent past, 
    while values near $1$ correspond to an almost uniform weighting over the 
    whole history.
\end{definition}

\begin{definition}[Time-varying power-law memory kernel] \label{def:memory_kernel}
    Let $\varrho \in \mathcal{C}(I)$ be a time-varying memory exponent as defined in 
    Definition~\ref{def:time_varying_exponent}.  For $0 \le \tau < t \le T$ we define 
    the associated \textbf{time-varying power-law memory kernel}
    \begin{equation} \label{eq:memory_kernel}
        K_{\varrho}(t,\tau) := \frac{(t-\tau)^{\varrho(t)-1}}{\Gamma(\varrho(t))}.
    \end{equation}
    For completeness we set $K_{\varrho}(t,\tau) := 0$ whenever $\tau \ge t$.
\end{definition}

\begin{remark}[Causal structure and relation to existing work] \label{rem:causal_structure}
    The kernel \eqref{eq:memory_kernel} is \textit{causal} (it vanishes for future times 
    $\tau \ge t$) and exhibits a power-law dependence on the elapsed time 
    $s = t - \tau$.  In the definition \eqref{eq:memory_kernel}, the exponent 
    $\varrho(t) - 1$ depends explicitly on the upper limit $t$.  In some 
    approaches to variable-order differential operators, the order function may be 
    evaluated at different points; for example, the works 
    \cite{coimbra2003mechanics,umarov2009variable,ramirez2010selection} consider 
    operators with piecewise constant or otherwise specified variable order.  
    The formulation presented here, where the memory exponent depends on the current 
    time $t$, provides one particular way to incorporate time-varying memory 
    characteristics into a non-local operator.
\end{remark}

\begin{lemma}[Integral formula for the time-varying memory kernel] \label{lem:kernel_integral_formula}
    Let $\varrho \in \mathcal{C}(I)$ be a time-varying memory exponent and let 
    $K_{\varrho}$ be the associated kernel defined in \eqref{eq:memory_kernel}.  
    For any $t \in (0,T]$,
    \begin{equation} \label{eq:kernel_normalization}
        \int_{0}^{t} K_{\varrho}(t,\tau)\, d\tau = \frac{t^{\varrho(t)}}{\Gamma(\varrho(t)+1)}.
    \end{equation}
    Consequently, for each fixed $t$, the kernel $K_{\varrho}(t,\cdot)$ is integrable 
    on $[0,t]$ and its total mass is a continuous function of $t$.
\end{lemma}

\begin{proof}
    Substituting the definition \eqref{eq:memory_kernel} and noting that $\varrho(t)$ 
    depends only on the upper limit $t$ (hence is constant with respect to the integration 
    variable $\tau$), we obtain
    \[
    \int_{0}^{t} K_{\varrho}(t,\tau)\, d\tau 
    = \frac{1}{\Gamma(\varrho(t))} \int_{0}^{t} (t-\tau)^{\varrho(t)-1}\, d\tau.
    \]
    The change of variable $u = t - \tau$ yields
    \[
    \int_{0}^{t} (t-\tau)^{\varrho(t)-1}\, d\tau = \int_{0}^{t} u^{\varrho(t)-1}\, du 
    = \frac{t^{\varrho(t)}}{\varrho(t)}.
    \]
    Using the functional equation $\Gamma(\varrho(t)+1) = \varrho(t)\,\Gamma(\varrho(t))$ 
    gives \eqref{eq:kernel_normalization}.  The continuity of the total mass as a function of 
    $t$ follows from the continuity of $\varrho$ and of the Gamma function on $(0,1]$.
\end{proof}

\begin{remark}[Special choices of the memory exponent] \label{rem:special_exponents}
    Several particular cases of $\varrho(\cdot)$ are worth noting:
    \begin{itemize}
        \item \textbf{Time-independent exponent} $\varrho(t) \equiv \rho \in (0,1]$: 
              then $K_{\varrho}(t,\tau)$ reduces to the classical kernel 
              $(t-\tau)^{\rho-1}/\Gamma(\rho)$ appearing in Riemann--Liouville 
              fractional integrals of order $\rho$.  The study of fractional 
              operators and their generalizations, including various approaches to 
              memory effects, has been pursued in many directions; see, e.g., works 
              on fractional stochastic models \cite{hasan2025novel}, differential 
              equations with fixed memory length \cite{ledesma2023differential}, 
              and statistical systems with variable memory \cite{kobelev2003statistical}.
              
        \item \textbf{Uniform memory} $\varrho(t) \equiv 1$: the kernel becomes 
              identically $1$, corresponding to an ordinary averaging over the past 
              without any preferential weighting.
              
        \item \textbf{Asymptotically vanishing exponent}: if $\varrho(t) \to 0^{+}$ as 
              $t$ increases, the kernel becomes increasingly singular near $\tau = t$,
              which can be interpreted as the system placing heightened emphasis on the 
              immediate past.  The study of systems whose memory characteristics evolve 
              over time has been undertaken in various contexts; see, for instance, 
              investigations of dynamic systems with variable memory 
              \cite{kachhia2021electromagnetic, kumar2019gegenbauer}.
    \end{itemize}
    In the sequel we keep $\varrho(\cdot)$ as a general continuous function taking values 
    in $(0,1]$; this generality allows the description of systems whose memory properties 
    evolve in time.
\end{remark}

\subsection{Definition of the Memory-Weighted Velocity Operator} \label{sec:velocity_operator_def}

Based on the time-varying memory kernel introduced in the preceding section, we now define 
the main object of study in this work.  The operator we shall consider is designed to 
incorporate historical information in a non-local manner when measuring rates of change.

\begin{definition}[Memory-weighted velocity operator] \label{def:velocity_operator}
    Let $\alpha, \beta \in \mathcal{C}(I)$ be two time-varying memory exponents 
    as in Definition~\ref{def:time_varying_exponent}.  For a function $x \in \mathcal{C}^{1}(I)$, 
    we define a new function $\mathscr{V}_{\alpha,\beta}[x] : I \to \mathbb{R}$ by
    \begin{equation} \label{eq:velocity_operator}
        \mathscr{V}_{\alpha,\beta}[x](t) := 
        \begin{cases}
            \displaystyle
            \frac{\int_{0}^{t} K_{\alpha}(t,\tau)\,\bigl[x(t)-x(\tau)\bigr]\,d\tau}
                 {\int_{0}^{t} K_{\beta}(t,\tau)\,(t-\tau)\,d\tau}, & t \in (0,T], \\[12pt]
            \dot{x}(0), & t = 0,
        \end{cases}
    \end{equation}
    where $K_{\alpha}$ and $K_{\beta}$ are the memory kernels corresponding to 
    $\alpha$ and $\beta$ respectively, as defined in \eqref{eq:memory_kernel}.  
    At the initial time $t=0$, the value is taken to be the ordinary derivative 
    $\dot{x}(0)$; this choice reflects the absence of historical information at 
    the starting point.
    
    The mapping 
    \[
    \mathscr{V}_{\alpha,\beta} : \mathcal{C}^{1}(I) \longrightarrow \mathcal{F}(I)
    \]
    is called the \textbf{memory-weighted velocity operator} with memory exponents 
    $\alpha$ and $\beta$.  Here $\mathcal{F}(I)$ denotes the set of all 
    real-valued functions on $I$.
\end{definition}

\begin{remark}[Interpretation and well-definedness] \label{rem:interpretation_well_defined}
    The definition \eqref{eq:velocity_operator} provides a generalized notion of 
    rate of change that incorporates historical information in a weighted manner.  
    
    \textbf{Well-definedness.} For $t \in (0,T]$, the denominator is strictly 
    positive because $K_{\beta}(t,\tau) > 0$ for all $\tau \in [0,t)$ and 
    $t-\tau > 0$ on this interval.  An explicit computation gives
    \[
     \int_{0}^{t} K_{\beta}(t,\tau)(t-\tau)\,d\tau 
    = \frac{1}{\Gamma(\beta(t))} \int_{0}^{t} (t-\tau)^{\beta(t)}\, d\tau 
    = \frac{t^{\beta(t)+1}}{(\beta(t)+1)\Gamma(\beta(t))} > 0.
    \]
    The numerator is also well-defined.  For a fixed $t \in (0,T]$, consider the integrand 
$F_t(\tau) := K_{\alpha}(t,\tau)\bigl[x(t)-x(\tau)\bigr]$ for $\tau \in [0,t)$.  

Since $\alpha(t) \in (0,1]$, the kernel $K_{\alpha}(t,\tau) = (t-\tau)^{\alpha(t)-1}/\Gamma(\alpha(t))$ 
has at most an integrable singularity as $\tau \to t^{-}$; indeed, 
$K_{\alpha}(t,\cdot) \in L^1[0,t]$ because $\alpha(t)-1 > -1$.

The factor $x(t)-x(\tau)$ is continuous on the compact interval $[0,t]$ and therefore 
bounded.  Moreover, the stronger estimate $|x(t)-x(\tau)| \leq L |t-\tau|$ holds for 
some $L > 0$ and all $\tau \in [0,t]$ because $x \in \mathcal{C}^{1}(I)$.  

Consequently, there exists a constant $C>0$ such that for all $\tau$ sufficiently close to $t$,  
$|F_t(\tau)| \leq C (t-\tau)^{\alpha(t)}$.  Since $\alpha(t) > 0$, this bound guarantees 
the integrability of $F_t$ on $[0,t]$.
    
    \textbf{Interpretation.} The operator $\mathscr{V}_{\alpha,\beta}$ provides a non‑local 
generalization of the derivative by forming the quotient of the accumulated historical 
change of $x$ and a weighted measure of elapsed time, both constructed from the 
same history.

\begin{itemize}
    \item \textbf{Numerator.} 
          The quantity $N_{\alpha}(t;x) := \int_{0}^{t} K_{\alpha}(t,\tau)\,\bigl[x(t)-x(\tau)\bigr]\,d\tau$
          accumulates past increments of $x$, weighted by the memory kernel $K_{\alpha}$.  
          The weighting emphasizes recent history more strongly when $\alpha(t)$ is small, 
          whereas a value of $\alpha(t)$ close to $1$ corresponds to a nearly uniform 
          weighting over the entire interval $[0,t]$.

    \item \textbf{Denominator.} 
          The quantity $D_{\beta}(t) := \int_{0}^{t} K_{\beta}(t,\tau)\,(t-\tau)\,d\tau$
          provides a weighted measure of the elapsed time.  It supplies a time scale 
          against which the accumulated change is measured.  An explicit computation gives 
          $D_{\beta}(t) = t^{\beta(t)+1}/[(\beta(t)+1)\Gamma(\beta(t))]$.

    \item \textbf{Two independent exponents.}
      Using separate exponents $\alpha$ and $\beta$ permits independent control 
      over two memory aspects: the weighting of past changes (governed by $\alpha$) 
      and the weighting of elapsed time (governed by $\beta$).  
      This separation provides additional flexibility in modeling systems where 
      the memory effects on the state and on the time scale may differ.
      In the symmetric case $\alpha=\beta$, both aspects obey the same memory law.

    \item \textbf{Initial condition.} 
          At $t=0$ no historical information is available; the definition therefore 
          prescribes $\mathscr{V}_{\alpha,\beta}[x](0) = \dot{x}(0)$, which coincides 
          with the classical derivative at the initial instant.
\end{itemize}
\end{remark}


Having defined the memory-weighted velocity operator, we now obtain a simplified explicit representation that will facilitate the subsequent analysis.

\begin{corollary}[Simplified expression] \label{cor:simplified_form}
    Let $\alpha, \beta \in \mathcal{C}(I)$ be time-varying memory exponents and 
    $x \in \mathcal{C}^{1}(I)$.  For every $t \in (0,T]$,
    \begin{equation} \label{eq:simplified_operator}
        \mathscr{V}_{\alpha,\beta}[x](t) = 
        \frac{(\beta(t)+1)\,\Gamma(\beta(t))}
             {t^{\beta(t)+1}\,\Gamma(\alpha(t))}
        \int_{0}^{t} (t-\tau)^{\alpha(t)-1} 
        \bigl[x(t)-x(\tau)\bigr] \, d\tau.
    \end{equation}
\end{corollary}

\begin{proof}
    From Definition~\ref{def:velocity_operator}, for $t \in (0,T]$,
    \[
    \mathscr{V}_{\alpha,\beta}[x](t) = 
    \frac{\displaystyle \int_{0}^{t} K_{\alpha}(t,\tau)\bigl[x(t)-x(\tau)\bigr] \, d\tau}
         {\displaystyle \int_{0}^{t} K_{\beta}(t,\tau)\,(t-\tau) \, d\tau}.
    \]
    The denominator has already been evaluated explicitly (see Remark~\ref{rem:interpretation_well_defined}):
    \[
    \int_{0}^{t} K_{\beta}(t,\tau)\,(t-\tau) \, d\tau 
    = \frac{t^{\beta(t)+1}}{(\beta(t)+1)\Gamma(\beta(t))}.
    \]
    Substituting this expression together with the explicit form 
    $K_{\alpha}(t,\tau) = (t-\tau)^{\alpha(t)-1}/\Gamma(\alpha(t))$ yields
    \begin{align*}
        \mathscr{V}_{\alpha,\beta}[x](t) 
        &= \frac{\displaystyle \frac{1}{\Gamma(\alpha(t))} 
                 \int_{0}^{t} (t-\tau)^{\alpha(t)-1} 
                 \bigl[x(t)-x(\tau)\bigr] \, d\tau}
               {\displaystyle \frac{t^{\beta(t)+1}}{(\beta(t)+1)\Gamma(\beta(t))}} \\[4pt]
        &= \frac{(\beta(t)+1)\,\Gamma(\beta(t))}
                 {t^{\beta(t)+1}\,\Gamma(\alpha(t))}
           \int_{0}^{t} (t-\tau)^{\alpha(t)-1} 
           \bigl[x(t)-x(\tau)\bigr] \, d\tau,
    \end{align*}
    which coincides precisely with \eqref{eq:simplified_operator}.
\end{proof}

\begin{proposition}[Linearity of the memory-weighted velocity operator] \label{prop:linearity}
    Let $\alpha, \beta \in \mathcal{C}(I)$ be time-varying memory exponents.  
    For any functions $x_1, x_2 \in \mathcal{C}^{1}(I)$ and any scalars 
    $c_1, c_2 \in \mathbb{R}$,
    \begin{equation} \label{eq:linearity}
        \mathscr{V}_{\alpha,\beta}[c_1 x_1 + c_2 x_2] = 
        c_1 \mathscr{V}_{\alpha,\beta}[x_1] + c_2 \mathscr{V}_{\alpha,\beta}[x_2].
    \end{equation}
    That is, for fixed exponents $\alpha$ and $\beta$, the mapping 
    $\mathscr{V}_{\alpha,\beta} : \mathcal{C}^{1}(I) \to \mathcal{F}(I)$ is a linear operator.
\end{proposition}

\begin{proof}
    Set $x = c_1 x_1 + c_2 x_2$.  For $t \in (0,T]$, employing the simplified expression 
    \eqref{eq:simplified_operator},
    \begin{align*}
        \mathscr{V}_{\alpha,\beta}[x](t) 
        &= \frac{(\beta(t)+1)\,\Gamma(\beta(t))}
                 {t^{\beta(t)+1}\,\Gamma(\alpha(t))}
           \int_{0}^{t} (t-\tau)^{\alpha(t)-1} 
           \bigl[x(t)-x(\tau)\bigr] \, d\tau \\
        &= \frac{(\beta(t)+1)\,\Gamma(\beta(t))}
                 {t^{\beta(t)+1}\,\Gamma(\alpha(t))}
           \int_{0}^{t} (t-\tau)^{\alpha(t)-1} 
           \bigl[c_1\bigl(x_1(t)-x_1(\tau)\bigr) 
               + c_2\bigl(x_2(t)-x_2(\tau)\bigr)\bigr] \, d\tau.
    \end{align*}
    By the linearity of the Riemann integral,
    \begin{align*}
        \int_{0}^{t} (t-\tau)^{\alpha(t)-1} & 
        \bigl[c_1\bigl(x_1(t)-x_1(\tau)\bigr) 
            + c_2\bigl(x_2(t)-x_2(\tau)\bigr)\bigr] \, d\tau \\
        &= c_1 \int_{0}^{t} (t-\tau)^{\alpha(t)-1} 
                 \bigl[x_1(t)-x_1(\tau)\bigr] \, d\tau \\
        &\quad + c_2 \int_{0}^{t} (t-\tau)^{\alpha(t)-1} 
                 \bigl[x_2(t)-x_2(\tau)\bigr] \, d\tau.
    \end{align*}
    Substituting this decomposition and invoking \eqref{eq:simplified_operator} once more yields
    \[
    \mathscr{V}_{\alpha,\beta}[x](t) = 
    c_1 \mathscr{V}_{\alpha,\beta}[x_1](t) + c_2 \mathscr{V}_{\alpha,\beta}[x_2](t).
    \]
    
    At the initial instant $t=0$, the equality follows directly from the definition:
    \[
    \mathscr{V}_{\alpha,\beta}[x](0) = \dot{x}(0) 
    = c_1 \dot{x}_1(0) + c_2 \dot{x}_2(0)
    = c_1 \mathscr{V}_{\alpha,\beta}[x_1](0) + c_2 \mathscr{V}_{\alpha,\beta}[x_2](0).
    \]
    
    Since the identity holds for every $t \in I$, the linearity relation \eqref{eq:linearity} 
    is established.
\end{proof}

\begin{remark}[Origin of the linearity] 
    \label{rem:linearity_consequence}
    Proposition~\ref{prop:linearity} confirms that for prescribed memory exponents 
    $\alpha$ and $\beta$, the mapping $\mathscr{V}_{\alpha,\beta} : \mathcal{C}^{1}(I) \to \mathcal{F}(I)$ 
    constitutes a linear operator.  This linear character originates directly from the 
    integral structure of the operator, which consists of a pointwise multiplication 
    by a coefficient depending solely on $t$, followed by a linear integral transform 
    applied to the difference $x(t)-x(\cdot)$.
\end{remark}

\subsection{Continuous dependence on the memory exponents}\label{subsec:continuous_dependence}

After introducing the memory-weighted velocity operator $\mathscr{V}_{\alpha,\beta}$ (Definition~\ref{def:velocity_operator}) and verifying its elementary properties such as linearity (Proposition~\ref{prop:linearity}), we now investigate how the operator behaves when the memory exponents are provided as sequences $\{\alpha_n\}, \{\beta_n\}$ rather than fixed functions.  

The following theorem demonstrates that if these sequences converge uniformly to limits $\alpha$ and $\beta$, then the corresponding operator sequence $\mathscr{V}_{\alpha_n,\beta_n}$ converges locally uniformly to $\mathscr{V}_{\alpha,\beta}$.  This result establishes that the velocity depends continuously on the memory exponents with respect to the topology of uniform convergence.

\begin{theorem}[Continuous dependence of the memory-weighted velocity operator on the memory exponents]\label{thm:continuous_dependence}
Let $\{\alpha_n\}_{n\geq1}, \{\beta_n\}_{n\geq1}, \alpha, \beta$ be time-varying memory exponent 
functions; that is, each function belongs to $\mathcal{C}(I)$ and takes values in $(0,1]$.  
Assume that:
\begin{enumerate}
    \item $\alpha_n$ and $\beta_n$ converge uniformly on $I$ to $\alpha$ and $\beta$, respectively; i.e.,
    \[
    \lim_{n\to\infty} \sup_{t\in I} |\alpha_n(t)-\alpha(t)| = 0,\qquad
    \lim_{n\to\infty} \sup_{t\in I} |\beta_n(t)-\beta(t)| = 0.
    \]
    \item There exist constants $0 < m \leq M \leq 1$ such that for all $n\geq1$ and all $t\in I$,
    \[
    m \leq \alpha_n(t), \beta_n(t), \alpha(t), \beta(t) \leq M.
    \]
\end{enumerate}
Then for any fixed function $x \in \mathcal{C}^1(I)$, the memory-weighted velocity operators converge 
locally uniformly on $(0,T]$: for every $\varepsilon \in (0,T)$,
\begin{equation}\label{eq:conv_main}
\sup_{t\in[\varepsilon,T]} 
\bigl| \mathscr{V}_{\alpha_n,\beta_n}[x](t) - \mathscr{V}_{\alpha,\beta}[x](t) \bigr| \to 0 \quad (n\to\infty).
\end{equation}
Equivalently,
\[
\mathscr{V}_{\alpha_n,\beta_n}[x] \rightrightarrows \mathscr{V}_{\alpha,\beta}[x] \quad \text{on } [\varepsilon,T] \ (n\to\infty).
\]
\end{theorem}

\begin{remark}
The convergence stated in Theorem~\ref{thm:continuous_dependence} is sometimes described as 
``locally uniform convergence on $(0,T]$'' or ``uniform convergence on every closed subinterval 
$[\varepsilon, T]$ with $\varepsilon > 0$.''  It reflects the fact that the operator 
$\mathscr{V}_{\alpha,\beta}[x]$ depends continuously on the memory exponents in a topology 
that respects the uniform convergence of the exponents.
\end{remark}

\begin{proof}
We divide the proof into five systematic steps.  Throughout we denote $I = [0,T]$ and fix 
an arbitrary $\varepsilon \in (0,T)$.  All statements concerning uniform convergence in the sequel 
refer to the closed subinterval $[\varepsilon, T]$.

\vspace{0.5em}
\noindent
\textbf{Step 1: Simplified representation and notation.}
From Corollary~\ref{cor:simplified_form}, for any $t \in (0,T]$ we have
\begin{equation} \label{eq:V_simplified_start}
\mathscr{V}_{\alpha,\beta}[x](t) = 
\frac{(\beta(t)+1)\Gamma(\beta(t))}
     {t^{\beta(t)+1}\Gamma(\alpha(t))}
\int_0^t (t-\tau)^{\alpha(t)-1} \bigl[x(t)-x(\tau)\bigr] \, d\tau.
\end{equation}
To streamline the presentation we introduce, for $t \in (0,T]$ (the condition $t>0$ ensures 
that the denominator $t^{\beta(t)+1}$ is non‑zero),
\[
A_n(t) := \frac{(\beta_n(t)+1)\Gamma(\beta_n(t))}
               {t^{\beta_n(t)+1}\Gamma(\alpha_n(t))},\qquad
A(t)   := \frac{(\beta(t)+1)\Gamma(\beta(t))}
               {t^{\beta(t)+1}\Gamma(\alpha(t))},
\]
and
\[
I_n(t) := \int_0^t (t-\tau)^{\alpha_n(t)-1} \bigl[x(t)-x(\tau)\bigr] \, d\tau,\qquad
I(t)   := \int_0^t (t-\tau)^{\alpha(t)-1} \bigl[x(t)-x(\tau)\bigr] \, d\tau.
\]
Thus for every $t \in (0,T]$,
\[
\mathscr{V}_{\alpha_n,\beta_n}[x](t) = A_n(t) I_n(t),\qquad
\mathscr{V}_{\alpha,\beta}[x](t) = A(t) I(t).
\]
Our aim is to show that the product $A_n I_n$ converges uniformly to $A I$ on $[\varepsilon, T]$.

\vspace{0.5em}
\noindent
\textbf{Step 2: Uniform convergence of the coefficient functions $A_n(t)$.}

\textbf{2.1. Handling the Gamma function terms.}  
By hypothesis (2) all function values lie in the compact interval $[m,M] \subset (0,1]$.  
The Gamma function $\Gamma : (0,\infty) \to \mathbb{R}$ is continuous on $[m,M]$ and therefore 
uniformly continuous on this compact set.  Consequently, for any prescribed $\vartheta_1 > 0$ there exists 
$\delta_1 > 0$ such that for all $z_1, z_2 \in [m,M]$ with $|z_1 - z_2| < \delta_1$, 
$|\Gamma(z_1) - \Gamma(z_2)| < \vartheta_1$.

Because $\alpha_n \rightrightarrows \alpha$ on $I$, we can select $N^{(1)}$ such that for all $n > N^{(1)}$,
\[
\sup_{t\in I} |\alpha_n(t) - \alpha(t)| < \delta_1.
\]
For any $t \in [\varepsilon, T]$ the numbers $\alpha_n(t)$ and $\alpha(t)$ belong to $[m,M]$ 
and satisfy $|\alpha_n(t) - \alpha(t)| < \delta_1$; hence
\[
|\Gamma(\alpha_n(t)) - \Gamma(\alpha(t))| < \vartheta_1.
\]
This establishes the uniform convergence $\Gamma \circ \alpha_n \rightrightarrows \Gamma \circ \alpha$ on $[\varepsilon, T]$.  
An identical argument yields $\Gamma \circ \beta_n \rightrightarrows \Gamma \circ \beta$ on $[\varepsilon, T]$.

\textbf{2.2. Handling the power terms $t^{\beta_n(t)+1}$.}  
Since $0 < m \leq M \leq 1$ by hypothesis, we have $\beta_n(t)+1 \in [m+1, M+1] \subset [1,2]$.  
Consider the two‑variable function
\[
\Phi(t,a) = t^a = e^{a\ln t}, \qquad (t,a) \in [\varepsilon, T] \times [m+1, M+1].
\]
Because $t \geq \varepsilon > 0$, $\ln t$ is bounded, so $\Phi$ is continuous on the compact 
set $[\varepsilon, T] \times [m+1, M+1]$ and consequently uniformly continuous there.

Given $\vartheta_2 > 0$, uniform continuity provides a $\delta_2 > 0$ such that whenever 
$|t_1 - t_2| < \delta_2$ and $|a_1 - a_2| < \delta_2$,
\[
|\Phi(t_1, a_1) - \Phi(t_2, a_2)| < \vartheta_2.
\]

Fix $t \in [\varepsilon, T]$ and take $t_1 = t_2 = t$, 
$a_1 = \beta_n(t)+1$, $a_2 = \beta(t)+1$.  
Since $\beta_n \rightrightarrows \beta$ on $I$, there exists $N^{(2)}$ such that for all $n > N^{(2)}$,
\[
\sup_{t\in [\varepsilon,T]} |\beta_n(t) - \beta(t)| < \delta_2,
\]
and consequently $|a_1 - a_2| < \delta_2$.  Then
\[
|t^{\beta_n(t)+1} - t^{\beta(t)+1}| = |\Phi(t,\beta_n(t)+1) - \Phi(t,\beta(t)+1)| < \vartheta_2.
\]
Thus $t^{\beta_n(\cdot)+1} \rightrightarrows t^{\beta(\cdot)+1}$ on $[\varepsilon, T]$.

\textbf{2.3. Convergence of the individual factors of $A_n(t)$.}  
We decompose $A_n(t)$ into four factors:
\[
A_n(t) = \underbrace{(\beta_n(t)+1)}_{B_n(t)}\,
        \underbrace{\Gamma(\beta_n(t))}_{G_n(t)}\,
        \underbrace{\frac{1}{t^{\beta_n(t)+1}}}_{P_n(t)}\,
        \underbrace{\frac{1}{\Gamma(\alpha_n(t))}}_{H_n(t)}.
\]
Define $B(t), G(t), P(t), H(t)$ analogously.

\begin{itemize}
    \item $B_n(t) = \beta_n(t)+1$:  From $\beta_n \rightrightarrows \beta$ and the continuity 
          of addition, $B_n \rightrightarrows B$ on $[\varepsilon, T]$.
    
    \item $G_n(t) = \Gamma(\beta_n(t))$:  Step 2.1 yields $\Gamma\circ\beta_n \rightrightarrows \Gamma\circ\beta$, 
          i.e., $G_n \rightrightarrows G$ on $[\varepsilon, T]$.
    
    \item $P_n(t) = \dfrac{1}{t^{\beta_n(t)+1}}$:  Consider its reciprocal $Q_n(t) = t^{\beta_n(t)+1}$.
          Step 2.2 has established the uniform convergence $Q_n \rightrightarrows Q$ on $[\varepsilon, T]$, 
          where $Q(t)=t^{\beta(t)+1}$.
          
          On the interval $[\varepsilon, T]$, because $\beta(t) \in [m, M]$, we have the uniform positive lower bound
          \[
          Q(t) = t^{\beta(t)+1} \geq \min\{\varepsilon^{M+1}, \varepsilon^{m+1}\} =: \eta_\varepsilon > 0.
          \]
          By the uniform convergence $Q_n \rightrightarrows Q$, there exists $N^{(3)}$ such that for all $n > N^{(3)}$,
          \[
          \sup_{t\in [\varepsilon,T]} |Q_n(t) - Q(t)| < \frac{\eta_\varepsilon}{2}.
          \]
          Consequently, for any $t \in [\varepsilon, T]$ and any $n > N^{(3)}$,
          \[
          Q_n(t) \geq Q(t) - |Q_n(t)-Q(t)| > \eta_\varepsilon - \frac{\eta_\varepsilon}{2} = \frac{\eta_\varepsilon}{2} > 0.
          \]
          
          \textbf{Upper bound estimation.}  To obtain a common compact interval containing all values of 
          $Q_n(t)$ (for $n > N^{(3)}$) and $Q(t)$, we observe:
          \[
          Q(t) = t^{\beta(t)+1} \leq 
          \begin{cases}
          1, & \text{if } T \leq 1 \ (\text{since } t \leq T \leq 1 \text{ and } \beta(t)+1 \geq 1), \\[4pt]
          T^{M+1}, & \text{if } T > 1 \ (\text{since } t \leq T \text{ and } \beta(t)+1 \leq M+1).
          \end{cases}
          \]
          Therefore
          \[
          Q(t) \leq \max\{1, T^{M+1}\} \quad \text{for all } t \in [\varepsilon, T].
          \]
          For $n > N^{(3)}$, using the uniform convergence estimate,
          \[
          Q_n(t) < Q(t) + \frac{\eta_\varepsilon}{2} \leq \max\{1, T^{M+1}\} + \frac{\eta_\varepsilon}{2}.
          \]
          
          Hence, for all $n > N^{(3)}$ and all $t \in [\varepsilon, T]$, the values $Q_n(t)$ and $Q(t)$ lie in the compact interval
          \[
          J_\varepsilon := \Bigl[\,\frac{\eta_\varepsilon}{2},\; \max\{1, T^{M+1}\} + \frac{\eta_\varepsilon}{2}\Bigr].
          \]
          
          The function $f(y) = 1/y$ is continuous on the compact interval $J_\varepsilon$ and therefore uniformly continuous there.
          A uniformly continuous function preserves uniform convergence under composition, which yields
          \[
          P_n = f \circ Q_n \rightrightarrows f \circ Q = P \quad \text{on } [\varepsilon, T] \ (n \to \infty).
          \]
    
    \item $H_n(t) = \dfrac{1}{\Gamma(\alpha_n(t))}$:  Step 2.1 gives $\Gamma\circ\alpha_n \rightrightarrows \Gamma\circ\alpha$.  
          Define
          \[
          \Gamma_{\min} := \min_{z \in [m,M]} \Gamma(z) > 0, \qquad
          \Gamma_{\max} := \max_{z \in [m,M]} \Gamma(z).
          \]
          Since $\alpha_n(t), \alpha(t) \in [m,M]$ for all $t \in [\varepsilon, T]$, we have
          \[
          \Gamma(\alpha_n(t)), \Gamma(\alpha(t)) \in [\Gamma_{\min}, \Gamma_{\max}] \quad \text{for all } t \in [\varepsilon, T].
          \]
          The function $g(y) = 1/y$ is uniformly continuous on the compact interval $[\Gamma_{\min}, \Gamma_{\max}]$;  
          consequently, by composition with the uniformly convergent sequence $\Gamma\circ\alpha_n$, we obtain
          \[
          H_n = g \circ (\Gamma\circ\alpha_n) \rightrightarrows g \circ (\Gamma\circ\alpha) = H \quad \text{on } [\varepsilon, T].
          \]
\end{itemize}

\textbf{2.4. Uniform convergence and boundedness of $A_n(t)$.}  
Each of the four factor sequences $B_n, G_n, P_n, H_n$ converges uniformly on $[\varepsilon, T]$, as established in Step 2.3.  
Moreover, these sequences are uniformly bounded on $[\varepsilon, T]$:

Using the constants $\Gamma_{\min}$ and $\Gamma_{\max}$ introduced in the analysis of $H_n(t)$, we have the following uniform bounds:
\begin{itemize}
    \item $|B_n(t)| = \beta_n(t)+1 \leq M+1$  (since $\beta_n(t) \leq M$);
    \item $|G_n(t)| = \Gamma(\beta_n(t)) \leq \Gamma_{\max}$  (by definition of $\Gamma_{\max}$);
    \item $|P_n(t)| = t^{-\beta_n(t)-1} \leq p_\varepsilon := \max\{\varepsilon^{-M-1},\varepsilon^{-m-1}\}$  
          (because $t \geq \varepsilon$ and $\beta_n(t) \in [m,M]$);
    \item $|H_n(t)| = \dfrac{1}{\Gamma(\alpha_n(t))} \leq \dfrac{1}{\Gamma_{\min}}$  
          (since $\Gamma(\alpha_n(t)) \geq \Gamma_{\min} > 0$).
\end{itemize}

Define the constant
\[
C_A := (M+1) \cdot \frac{\Gamma_{\max}}{\Gamma_{\min}} \cdot p_\varepsilon.
\]

\textbf{Boundedness of both $A_n(t)$ and the limit $A(t)$.}  
For every $n\geq 1$ and every $t\in[\varepsilon,T]$, the product bound yields
\[
|A_n(t)| = |B_n(t)G_n(t)P_n(t)H_n(t)| \leq C_A.
\]

Now, by Lemma~\ref{lem:product_uniform_convergence} (specifically, Step 1 of its proof), the uniform boundedness 
of each factor sequence together with its uniform convergence implies that the corresponding limit function 
satisfies the same bound.  Consequently, for the limit functions $B(t), G(t), P(t), H(t)$ we have:
\[
|B(t)| \leq M+1,\quad |G(t)| \leq \Gamma_{\max},\quad |P(t)| \leq p_\varepsilon,\quad |H(t)| \leq \frac{1}{\Gamma_{\min}}.
\]
Therefore,
\[
|A(t)| = |B(t)G(t)P(t)H(t)| \leq C_A \quad \text{for all } t \in [\varepsilon, T].
\]

Thus both $\{A_n\}$ and $A$ are uniformly bounded by the same constant $C_A$ on $[\varepsilon, T]$.

Finally, Lemma~\ref{lem:product_uniform_convergence} (applied with $k=4$ to the factor sequences 
$B_n, G_n, P_n, H_n$) asserts that the product of finitely many uniformly convergent, 
uniformly bounded sequences converges uniformly.  Hence,
\begin{equation} \label{eq:A_n_uniform_convergence}
A_n \rightrightarrows A \quad \text{on } [\varepsilon, T],
\end{equation}
or equivalently,
\[
\lim_{n\to\infty} \sup_{t\in[\varepsilon,T]} |A_n(t) - A(t)| = 0.
\]

\vspace{0.5em}
\noindent
\textbf{Step 3: Uniform convergence of the integral parts $I_n(t)$.}

\textbf{3.1. Lipschitz estimate.}  
Since $x \in \mathcal{C}^1(I)$ and $I$ is compact, $x$ is Lipschitz continuous on $I$.  
Let $L_x > 0$ be a Lipschitz constant for $x$, so that for all $t,\tau \in I$,
\[
|x(t)-x(\tau)| \leq L_x |t-\tau|.
\]

\textbf{3.2. Integral representation of the difference.}  
For $t \in [\varepsilon, T]$,
\begin{align*}
|I_n(t)-I(t)| 
&\leq \int_0^t \bigl| (t-\tau)^{\alpha_n(t)-1} - (t-\tau)^{\alpha(t)-1} \bigr|
       |x(t)-x(\tau)| \, d\tau \\
&\leq L_x \int_0^t \bigl| (t-\tau)^{\alpha_n(t)-1} - (t-\tau)^{\alpha(t)-1} \bigr|
       (t-\tau) \, d\tau.
\end{align*}
Perform the change of variable $u = t-\tau$ (hence $du = -d\tau$).  
When $\tau$ runs from $0$ to $t$, $u$ runs from $t$ to $0$; reversing the integration limits removes the minus sign, giving
\[
|I_n(t)-I(t)| \leq L_x \int_0^t \bigl| u^{\alpha_n(t)-1} - u^{\alpha(t)-1} \bigr| u \, du.
\]

\textbf{3.3. Algebraic simplification.}  
Observe the elementary identity
\[
u^{\alpha_n(t)} - u^{\alpha(t)} = u\bigl(u^{\alpha_n(t)-1} - u^{\alpha(t)-1}\bigr),
\]
which implies
\[
\bigl| u^{\alpha_n(t)-1} - u^{\alpha(t)-1} \bigr| \cdot u = \bigl| u^{\alpha_n(t)} - u^{\alpha(t)} \bigr|.
\]
Therefore
\begin{equation} \label{eq:I_diff_simplified}
|I_n(t)-I(t)| \leq L_x \int_0^t \bigl| u^{\alpha_n(t)} - u^{\alpha(t)} \bigr| \, du.
\end{equation}

\textbf{Remark (behavior at the left endpoint).}  
Because $\alpha_n(t),\alpha(t) \geq m > 0$, we have $\lim_{u\to 0^+} u^{\alpha_n(t)} = 0$ and 
$\lim_{u\to 0^+} u^{\alpha(t)} = 0$.  Defining the integrand at $u=0$ to be $0$ renders the function 
$u \mapsto u^{\alpha_n(t)} - u^{\alpha(t)}$ continuous on the whole interval $[0,t]$.  
Consequently the integral is well defined in both the Riemann and Lebesgue senses, and the point 
$u=0$ does not affect its value.

\textbf{3.4. Application of the mean‑value theorem.}  
For fixed $u > 0$ consider the function $h(p)=u^p = e^{p\ln u}$.  Its derivative is 
$h'(p)=u^p \ln u$.  By the mean‑value theorem, for any real numbers $p_1, p_2$ there exists 
$\xi$ between $p_1$ and $p_2$ such that
\begin{equation} \label{eq:mean_value_power}
u^{p_1} - u^{p_2} = h(p_1) - h(p_2) = h'(\xi)(p_1-p_2) = u^{\xi} \ln u \,(p_1-p_2).
\end{equation}
Take $p_1 = \alpha_n(t)$ and $p_2 = \alpha(t)$.  Then $\xi = \xi(u,t,n)$ satisfies
\[
\min\{\alpha_n(t),\alpha(t)\} \leq \xi \leq \max\{\alpha_n(t),\alpha(t)\} \subset [m,M],
\]
and consequently
\[
\bigl| u^{\alpha_n(t)} - u^{\alpha(t)} \bigr| 
= u^{\xi} |\ln u| \, |\alpha_n(t)-\alpha(t)|.
\]
\textbf{3.5. Uniform bound for the factor \(u^{\xi}|\ln u|\) on the region \([0,T]\times[m,M]\).}  
Define
\[
F(u,\xi) = 
\begin{cases}
u^{\xi} |\ln u|, & u > 0,\\
0, & u = 0,
\end{cases}
\qquad (u,\xi) \in D := [0,T]\times[m,M].
\]
We show that \(F\) is continuous and bounded on the compact set \(D\).

\textit{Continuity at \(u=0\):}  By Lemma~\ref{lem:log_power}, for any \(\delta > 0\) we have 
\(|\ln u| \leq \frac{1}{\delta e} u^{-\delta}\).  Taking \(\delta = \xi/2 > 0\) (since \(\xi \geq m > 0\)),
\[
u^{\xi} |\ln u| \leq u^{\xi} \cdot \frac{1}{\frac{\xi}{2} e} u^{-\xi/2} 
                 = \frac{2}{\xi e} u^{\xi/2}.
\]
Because \(\xi/2 > 0\), \(\lim_{u\to 0^+} u^{\xi/2} = 0\); hence \(\lim_{u\to 0^+} u^{\xi}|\ln u| = 0\).  
Thus the definition \(F(0,\xi)=0\) makes \(F\) continuous at \(u=0\).

Now we estimate \(F\) separately for \(u \in (0,1]\) and for \(u \in [1,T]\) (the latter only when \(T>1\)).

\textit{Case (i): \(u \in (0,1]\).}  Using Lemma~\ref{lem:log_power} with \(\delta = m/2 > 0\),
\[
|\ln u| \leq \frac{1}{\delta e} u^{-\delta} = \frac{2}{m e} u^{-m/2}.
\]
For \(u \in (0,1]\) and \(\xi \in [m,M]\),
\[
u^{\xi} |\ln u| \leq \frac{2}{m e} u^{\xi - m/2} \leq \frac{2}{m e} u^{m - m/2} 
                = \frac{2}{m e} u^{m/2}.
\]
Since \(m/2 > 0\), the function \(u \mapsto u^{m/2}\) is increasing on \((0,1]\), with 
\(\lim_{u\to 0^+} u^{m/2} = 0\) and \(1^{m/2}=1\); therefore \(u^{m/2} \leq 1\) for all \(u \in (0,1]\).  
Consequently,
\[
F(u,\xi) \leq \frac{2}{m e} \quad \text{for } u \in (0,1],\ \xi \in [m,M].
\]
Moreover, by the definition \(F(0,\xi)=0\), the same bound holds trivially at \(u=0\).

\textit{Case (ii): \(u \in [1, T]\) (when \(T > 1\)).}  By Lemma~\ref{lem:log_growth} with \(\alpha_0 = 1\),
\[
|\ln u| \leq u \quad \text{for all } u \geq 1.
\]
Hence for \(u \in [1,T]\) and \(\xi \in [m,M]\),
\[
u^{\xi} |\ln u| \leq u^{M} \cdot u = u^{M+1} \leq T^{M+1}.
\]

\textbf{Synthesis of the estimates.}  Distinguish two situations:
\begin{itemize}
    \item If \(T \leq 1\), then for all \((u,\xi) \in D\) we have \(u \in [0,1]\).  
          From Case (i) and the definition \(F(0,\xi)=0\), we obtain
          \[
          F(u,\xi) \leq \frac{2}{m e} \quad \text{for all } (u,\xi) \in D,
          \]
          and we may set \(C_F := \frac{2}{m e}\).
    
    \item If \(T > 1\), the region splits as \(D = ([0,1]\times[m,M]) \cup ([1,T]\times[m,M])\).  
          From the two cases we obtain
          \[
          F(u,\xi) \leq \frac{2}{m e} \ \text{for } u\in[0,1],\qquad
          F(u,\xi) \leq T^{M+1} \ \text{for } u\in[1,T].
          \]
          Hence for every \((u,\xi) \in D\),
          \[
          F(u,\xi) \leq \max\Bigl\{\frac{2}{m e},\ T^{M+1}\Bigr\},
          \]
          and we set \(C_F := \max\bigl\{\frac{2}{m e},\ T^{M+1}\bigr\}\).
\end{itemize}
For a unified notation we may write
\[
C_F := 
\begin{cases}
\dfrac{2}{m e}, & T \leq 1, \\[8pt]
\displaystyle\max\Bigl\{\frac{2}{m e},\ T^{M+1}\Bigr\}, & T > 1.
\end{cases}
\]
Thus
\[
F(u,\xi) = u^{\xi} |\ln u| \leq C_F \quad \text{for all } (u,\xi) \in D,
\]
showing that \(F\) is uniformly bounded on the compact set \(D\).

\textbf{3.6. Integral estimate and uniform convergence.}  
Inserting the bound from Step 3.4,
\[
|I_n(t)-I(t)| \leq L_x \int_0^t u^{\xi(u,t,n)} |\ln u| \, |\alpha_n(t)-\alpha(t)| \, du.
\]
For each fixed $t$ and $n$, the point $\xi(u,t,n)$ belongs to $[m,M]$ for all $u \in [0,t]$.  
Hence by the uniform bound established in Step 3.5,
\[
u^{\xi(u,t,n)} |\ln u| \leq C_F \quad \text{for all } u \in [0,t].
\]
Therefore
\[
|I_n(t)-I(t)| \leq L_x C_F |\alpha_n(t)-\alpha(t)| \int_0^t du
                = L_x C_F |\alpha_n(t)-\alpha(t)| \cdot t.
\]
Since $t \leq T$ on $[\varepsilon, T]$,
\[
|I_n(t)-I(t)| \leq L_x C_F T \, |\alpha_n(t)-\alpha(t)|.
\]
Taking suprema over $t \in [\varepsilon, T]$,
\[
\sup_{t\in[\varepsilon,T]} |I_n(t)-I(t)| 
\leq L_x C_F T \sup_{t\in[\varepsilon,T]} |\alpha_n(t)-\alpha(t)|
\leq L_x C_F T \sup_{t\in I} |\alpha_n(t)-\alpha(t)|.
\]
The right‑hand side tends to zero as $n\to\infty$ because $\alpha_n \rightrightarrows \alpha$ on $I$.  
Consequently,
\begin{equation} \label{eq:I_uniform_convergence}
\lim_{n\to\infty} \sup_{t\in[\varepsilon,T]} |I_n(t)-I(t)| = 0,
\end{equation}
i.e., $I_n \rightrightarrows I$ on $[\varepsilon, T]$.

\textbf{3.7. Uniform boundedness of $\{I_n\}$ and $I$.}  
From the Lipschitz estimate in Step 3.1,
\[
|I_n(t)| \leq L_x \int_0^t (t-\tau)^{\alpha_n(t)} \, d\tau.
\]
The change of variable $u = t-\tau$ yields
\[
\int_0^t (t-\tau)^{\alpha_n(t)} \, d\tau = \int_0^t u^{\alpha_n(t)} \, du.
\]
Since $\alpha_n(t) \in [m, M] \subset (0,1]$, Lemma~\ref{lem:power_linear_bound} gives 
$u^{\alpha_n(t)} \leq 1 + u$ for all $u \geq 0$.  Hence,
\[
\int_0^t u^{\alpha_n(t)} \, du \leq \int_0^t (1 + u) \, du = t + \frac{t^2}{2}.
\]
Because $t \in [\varepsilon, T]$, we have $t \leq T$ and $t^2 \leq T^2$; thus
\[
|I_n(t)| \leq L_x\Bigl(T + \frac{T^2}{2}\Bigr).
\]

The same argument applied to the limit exponent $\alpha(t)$ (which also lies in $[m,M]$) yields
\[
|I(t)| \leq L_x\Bigl(T + \frac{T^2}{2}\Bigr).
\]

Therefore we may set the uniform bound
\[
C_I := L_x\Bigl(T + \frac{T^2}{2}\Bigr),
\]
so that for every $n \geq 1$ and every $t \in [\varepsilon, T]$,
\[
|I_n(t)| \leq C_I,\qquad |I(t)| \leq C_I.
\]

\vspace{0.5em}
\noindent
\textbf{Step 4: Uniform convergence of the product $A_n I_n$.}

We have established the following three facts:
\begin{enumerate}
    \item $A_n \rightrightarrows A$ on $[\varepsilon, T]$ (equation \eqref{eq:A_n_uniform_convergence});
    \item $I_n \rightrightarrows I$ on $[\varepsilon, T]$ (equation \eqref{eq:I_uniform_convergence});
    \item The sequences $\{A_n\}, \{A\}$ are uniformly bounded by $C_A$, and 
          $\{I_n\}, \{I\}$ are uniformly bounded by $C_I$.
\end{enumerate}
Now for any $t \in [\varepsilon, T]$, using the elementary device of adding and subtracting 
$A_n(t)I(t)$,
\begin{align*}
\bigl| A_n(t)I_n(t) - A(t)I(t) \bigr| 
&= \bigl| A_n(t)I_n(t) - A_n(t)I(t) + A_n(t)I(t) - A(t)I(t) \bigr| \\
&\leq \bigl| A_n(t)I_n(t) - A_n(t)I(t) \bigr| 
     + \bigl| A_n(t)I(t) - A(t)I(t) \bigr| \\
&= |A_n(t)| \, |I_n(t)-I(t)| + |I(t)| \, |A_n(t)-A(t)| \\
&\leq C_A |I_n(t)-I(t)| + C_I |A_n(t)-A(t)|.
\end{align*}
Taking the supremum over all $t \in [\varepsilon, T]$ yields
\[
\sup_{t\in[\varepsilon,T]} \bigl| A_n(t)I_n(t) - A(t)I(t) \bigr| 
\leq C_A \sup_{t\in[\varepsilon,T]} |I_n(t)-I(t)| 
     + C_I \sup_{t\in[\varepsilon,T]} |A_n(t)-A(t)|.
\]
Both terms on the right‑hand side tend to zero as $n \to \infty$ owing to the uniform convergences 
established in Steps 2 and 3.  Consequently,
\begin{equation} \label{eq:product_uniform_convergence}
\lim_{n\to\infty} \sup_{t\in[\varepsilon,T]} \bigl| A_n(t)I_n(t) - A(t)I(t) \bigr| = 0,
\end{equation}
which is precisely the statement that $A_n I_n$ converges uniformly to $A I$ on $[\varepsilon, T]$.

\vspace{0.5em}
\noindent
\textbf{Step 5: Completion of the proof.}

Recall from Step 1 that for every $t \in (0,T]$,
\[
\mathscr{V}_{\alpha_n,\beta_n}[x](t) = A_n(t)I_n(t),\qquad
\mathscr{V}_{\alpha,\beta}[x](t) = A(t)I(t),
\]
while at the initial instant $t=0$ both operators are defined to equal the classical derivative 
$\dot{x}(0)$.  Since we are working on the subinterval $[\varepsilon, T]$ with $\varepsilon > 0$, 
the identities above are valid for all $t$ under consideration.  Therefore Step 4 yields directly
\begin{equation} \label{eq:final_convergence}
\lim_{n\to\infty} \sup_{t\in[\varepsilon,T]} 
\bigl| \mathscr{V}_{\alpha_n,\beta_n}[x](t) - \mathscr{V}_{\alpha,\beta}[x](t) \bigr| = 0.
\end{equation}

The positive number $\varepsilon$ was chosen arbitrarily in $(0,T)$; the proof demonstrates that 
for each such $\varepsilon$ the convergence is uniform on $[\varepsilon, T]$.  
Consequently, for any fixed function $x \in \mathcal{C}^1(I)$, the operator sequence 
$\mathscr{V}_{\alpha_n,\beta_n}[x]$ converges to $\mathscr{V}_{\alpha,\beta}[x]$ 
\textbf{uniformly on every compact interval of the form $[\varepsilon, T]$ with $0 < \varepsilon < T$}.

In standard terminology, this is exactly \textbf{local uniform convergence} on the half‑open 
interval $(0,T]$: for every compact subset $K \subset (0,T]$, the convergence is uniform on $K$.
(Indeed, any such compact $K$ is contained in some interval $[\varepsilon, T]$ with $\varepsilon > 0$, 
and the uniformity on $[\varepsilon, T]$ implies uniformity on $K$.)

This completes the proof of Theorem~\ref{thm:continuous_dependence}.  The theorem establishes that 
the memory‑weighted velocity operator depends continuously on its time‑varying memory exponents 
when the exponents are equipped with the topology of uniform convergence.
\end{proof}

\begin{remark}[Continuous dependence in the context of operator theory and nonlinear analysis]
Theorem~\ref{thm:continuous_dependence} establishes that the memory-weighted velocity operator \(\mathscr{V}_{\alpha,\beta}\) depends continuously on its memory exponents in the following precise sense: the mapping
\[
(\alpha,\beta) \mapsto \mathscr{V}_{\alpha,\beta}
\]
is continuous when the exponent functions are equipped with the topology of uniform convergence and the resulting operators are considered with the topology of local uniform convergence on \((0,T]\). This can be viewed as a form of **stable dependence of an operator family on its defining parameters**—a property of fundamental interest in both linear operator theory and the analysis of parameter-dependent nonlinear evolution equations.

From the perspective of functional analysis, the result ensures that the operator \(\mathscr{V}_{\alpha,\beta}\) behaves regularly under perturbations of its kernel exponents, which is essential for studying its spectral approximation, compactness properties, or embedding into operator ideals. Within the broader context of nonlinear analysis, this continuity underpins the well-posedness of equations involving memory-weighted velocities: if the memory exponents are subject to uncertainty, approximation, or adaptive evolution, the corresponding velocity fields remain close, thereby preserving the qualitative behavior of solutions.

The proof combines real-analytic estimates with operator-theoretic reasoning—decomposing \(\mathscr{V}_{\alpha,\beta}\) into a product of coefficient and integral components, and controlling their uniform convergence via boundedness and continuity arguments. This blend of concrete estimates and functional-analytic structure reflects the dual nature of the operator as both an integral transform and a weighted bounded linear map between function spaces, in the sense of Theorem~\ref{thm:weighted_boundedness}.

Beyond its theoretical value, the continuity established here provides a rigorous foundation for numerical and modeling practices: it justifies the use of discretized or experimentally inferred memory exponents, supports parameter identification procedures in inverse problems (where memory exponents are to be recovered from observational data), and facilitates the analysis of evolution equations in which memory characteristics themselves evolve over time—such as in viscoelastic materials with aging or transport processes with history-dependent diffusivity. These connections position \(\mathscr{V}_{\alpha,\beta}\) as a versatile object at the interface of operator theory and nonlinear dynamics.
\end{remark}

\subsection{The Uniform-Memory Regime: Asymptotic Recovery of the Classical Derivative}
\label{subsec:uniform_memory_case}

We now examine an important special case of the memory-weighted velocity operator 
introduced in Definition~\ref{def:velocity_operator}. When both memory exponents are 
identically equal to one, i.e., $\alpha(t)\equiv 1$ and $\beta(t)\equiv 1$, the 
memory kernels simplify substantially. Indeed, from \eqref{eq:memory_kernel} with 
$\varrho(t)=1$ and $\Gamma(1)=1$, both memory kernels reduce to the constant function one: $K_{\alpha}(t,\tau) = K_{\beta}(t,\tau) \equiv 1$ for all $0 \le \tau < t \le T$. Consequently, the general expression \eqref{eq:velocity_operator} 
reduces to a purely time-averaged form. This ``uniform-memory'' regime serves as a 
natural testing ground for the theory and provides a direct link to classical 
differential calculus.

\begin{theorem}[Asymptotic recovery of the classical derivative at the origin]
\label{thm:uniform_memory_limit}
Let $x \in \mathcal{C}^{1}(I)$. The uniform-memory velocity operator $\mathscr{V}_{1,1}$, obtained from 
Definition~\ref{def:velocity_operator} by setting $\alpha(t)\equiv\beta(t)\equiv 1$, is given by
\begin{equation}\label{eq:V11_form}
\mathscr{V}_{1,1}[x](t) = 
\begin{cases}
\displaystyle
\frac{\int_{0}^{t} \bigl[ x(t) - x(\tau) \bigr] \, d\tau}
     {\int_{0}^{t} (t - \tau) \, d\tau}, & t \in (0, T], \\[12pt]
\dot{x}(0), & t = 0.
\end{cases}
\end{equation}
For this operator, the following asymptotic relation holds:
\begin{equation}\label{eq:V11_limit}
\lim_{t \to 0^{+}} \mathscr{V}_{1,1}[x](t) = \dot{x}(0).
\end{equation}
\end{theorem}

\begin{proof}
The proof is organized into eight systematic steps, each providing a detailed derivation together with the necessary justifications to ensure logical completeness.

\vspace{0.5em}
\noindent
\textbf{Step 1: Simplifying the expression.}

First compute the denominator exactly:
\begin{align}
\int_{0}^{t} (t - \tau) \, d\tau 
&= \Bigl[ t\tau - \frac{\tau^{2}}{2} \Bigr]_{\tau=0}^{\tau=t} 
 = t^{2} - \frac{t^{2}}{2} 
 = \frac{t^{2}}{2}. \label{eq:denominator_exact}
\end{align}
(Note that $t$ is fixed during the integration over $\tau$.)

Using this result, the operator $\mathscr{V}_{1,1}[x](t)$ defined in the theorem can be rewritten as
\begin{equation}
\label{eq:V11_simplified}
\mathscr{V}_{1,1}[x](t) = \frac{2}{t^{2}} \int_{0}^{t} \bigl[ x(t) - x(\tau) \bigr] \, d\tau, \qquad t \in (0, T].
\end{equation}
This compact form will facilitate the subsequent limit analysis.

\vspace{0.5em}
\noindent
\textbf{Step 2: Application of the first mean‑value theorem for integrals (representation via an intermediate point).}

We now analyze the numerator $\int_{0}^{t} [x(t) - x(\tau)] \, d\tau$. It is essential to verify the conditions for applying the first mean‑value theorem for integrals.

Since $x \in \mathcal{C}^{1}(I)$, in particular $x$ is continuous on $I$. Hence for a fixed $t > 0$, the function
\[
f_{t}(\tau) := x(t) - x(\tau), \qquad \tau \in [0, t],
\]
is continuous on the closed interval $[0, t]$. The weight function is constant and equal to $1$, which is certainly integrable and does not change sign on $[0, t]$.

The precise statement of the \textbf{first mean‑value theorem for integrals} (for continuous integrands) is: if $F$ is continuous on $[a, b]$ and $w$ is integrable and does not change sign on $[a, b]$, then there exists $\xi \in [a, b]$ such that
\[
\int_{a}^{b} F(\tau) w(\tau) \, d\tau = F(\xi) \int_{a}^{b} w(\tau) \, d\tau.
\]

In our setting we take $[a, b] = [0, t]$, $F(\tau) = x(t) - x(\tau)$, and $w(\tau) \equiv 1$. Applying the theorem, we obtain a point $\xi(t) \in [0, t]$ (notice that $\xi(t)$ depends on $t$) satisfying
\begin{align}
\int_{0}^{t} \bigl[ x(t) - x(\tau) \bigr] \, d\tau 
&= \bigl[ x(t) - x(\xi(t)) \bigr] \int_{0}^{t} 1 \, d\tau \notag \\
&= \bigl[ x(t) - x(\xi(t)) \bigr] \cdot t. \label{eq:mean_value_numerator}
\end{align}
Here $\int_{0}^{t} 1 \, d\tau = t$ is precisely the length of the interval.

For later convenience we introduce the normalized (relative) position $\theta(t)$ defined by
\begin{equation}
\label{eq:theta_def}
\theta(t) := \frac{\xi(t)}{t}, \qquad \text{so that } \theta(t) \in [0, 1].
\end{equation}
Thus $\xi(t) = \theta(t) t$ and consequently $x(\xi(t)) = x(\theta(t) t)$. 
Using this notation, equation \eqref{eq:mean_value_numerator} takes the compact form
\[
\int_{0}^{t} \bigl[ x(t) - x(\tau) \bigr] \, d\tau = \bigl[ x(t) - x(\theta(t) t) \bigr] \cdot t.
\]

\vspace{0.5em}
\noindent
\textbf{Step 3: Expressing the memory velocity in terms of the intermediate point.}

Combining the result of Step 2, \eqref{eq:mean_value_numerator}, with the exact denominator $t^{2}/2$ from Step 1, we obtain from the expression for $\mathscr{V}_{1,1}[x](t)$:
\begin{align}
\mathscr{V}_{1,1}[x](t) 
&= \frac{\bigl[ x(t) - x(\xi(t)) \bigr] \cdot t}{t^{2}/2} \notag \\
&= \frac{x(t) - x(\xi(t))}{t/2}. \label{eq:V11_in_xi}
\end{align}
This representation exhibits the memory velocity as a difference quotient with respect to the intermediate point $\xi(t)$, where the effective increment is $t/2$ (half the length of the integration interval).
\vspace{0.5em}
\noindent

\textbf{Step 4: Applying the Lagrange mean‑value theorem to $x(t)-x(\xi(t))$.}

To analyze $x(t) - x(\xi(t))$ further we employ the Lagrange mean‑value theorem. Because $x \in \mathcal{C}^{1}(I)$, the theorem applies on any closed subinterval of $I$. We distinguish two possibilities according to the relative position of $\xi(t)$.

\textbf{Case 1: $\xi(t) < t$ (non‑degenerate interval).}  
Then $t - \xi(t) > 0$ and the interval $[\xi(t), t]$ has positive length. By the Lagrange mean‑value theorem, there exists at least one point $\zeta(t) \in (\xi(t), t)$ such that
\begin{equation}
\label{eq:lagrange_strict}
x(t) - x(\xi(t)) = \dot{x}(\zeta(t)) \cdot (t - \xi(t)).
\end{equation}
Here $\dot{x}$ denotes the derivative of $x$. Note that $\zeta(t)$ depends on $t$ and satisfies $\xi(t) < \zeta(t) < t$.

\textbf{Case 2: $\xi(t) = t$ (degenerate interval).}  
In this situation $x(t) - x(\xi(t)) = 0$ and $t - \xi(t) = 0$. The equality
\[
x(t) - x(\xi(t)) = \dot{x}(\zeta(t)) \cdot (t - \xi(t))
\]
holds for any choice of $\zeta(t) \in [0, t]$ because both sides are zero. For definiteness we choose $\zeta(t) = t$ (which coincides with $\xi(t)$ in this case).

\textbf{Remark: The case $\xi(t) > t$ is impossible.}  
Indeed, from the definition \eqref{eq:theta_def} we have $\theta(t) = \xi(t)/t \in [0, 1]$, which implies $\xi(t) \leq t$ for all $t > 0$.

In summary, we can always select a point $\zeta(t) \in [0, t]$ for which
\begin{equation}
\label{eq:lagrange_general}
x(t) - x(\xi(t)) = \dot{x}(\zeta(t)) \cdot (t - \xi(t)).
\end{equation}
Concretely:
\begin{itemize}
    \item when $\xi(t) < t$, the point $\zeta(t)$ is provided by the mean‑value theorem and satisfies $\xi(t) < \zeta(t) < t$;
    \item when $\xi(t) = t$, we have set $\zeta(t) = t$.
\end{itemize}
In either case $\zeta(t) \in [0, t]$.

\vspace{0.5em}
\noindent
\textbf{Step 5: Expressing the memory velocity using $\theta(t)$ and $\zeta(t)$.}

Insert \eqref{eq:lagrange_general} into \eqref{eq:V11_in_xi} and use the relation $\xi(t) = \theta(t) t$ from \eqref{eq:theta_def}:
\begin{align}
\mathscr{V}_{1,1}[x](t) 
&= \frac{\dot{x}(\zeta(t)) \cdot (t - \xi(t))}{t/2} \notag \\
&= 2 \dot{x}(\zeta(t)) \cdot \bigl(1 - \theta(t)\bigr). \label{eq:V11_theta_form}
\end{align}
This compact formula expresses the memory velocity as the product of three factors: 
the constant $2$, the derivative value $\dot{x}(\zeta(t))$ at the point $\zeta(t)$ provided by the Lagrange mean‑value theorem, and $1-\theta(t)$ where $\theta(t)=\xi(t)/t$ is the normalized position of the integral mean‑value point $\xi(t)$.

\vspace{0.5em}
\noindent
\textbf{Step 6: Employing the error‑control functions constructed in Appendix~\ref{app:error_control}.}

Recall from \eqref{eq:V11_simplified} that $\mathscr{V}_{1,1}[x](t) = \frac{2}{t^{2}} \int_{0}^{t} [x(t)-x(\tau)] \, d\tau$.

Set $a_0 := \dot{x}(0)$. By the differentiability of $x$ at $0$, we introduce the remainder function
\[
r(h) := x(h) - x(0) - a_0 h, \qquad h \in I.
\]

We now invoke the machinery developed in Appendix~\ref{app:error_control}:
\begin{enumerate}
    \item The function $r(h)$ defined above coincides precisely with the remainder function introduced in Definition~\ref{def:remainder_function_app}.
    
    \item Following Definition~\ref{def:epsilon_function_app}, we construct the error‑control function
          $\epsilon(s) := \sup_{0 < h \leq s} |r(h)/h|$ for $s>0$, with $\epsilon(0):=0$.
    
    \item For a fixed $t>0$, define $\epsilon_{t} := \sup_{0 \leq s \leq t} \epsilon(s)$ in accordance with Definition~\ref{def:epsilon_t_app}. Proposition~\ref{prop:epsilon_t_explicit_app} supplies the useful identification $\epsilon_{t} = \epsilon(t)$ together with the limit $\displaystyle\lim_{t\to 0^{+}} \epsilon_{t} = 0$.
    
    \item Lemma~\ref{lem:uniform_error_bound_app} furnishes the uniform bounds: for any $0 \leq \tau \leq t$,
          \begin{align}
          |r(\tau)| &\leq \epsilon_{t} \, \tau, \label{eq:bound_r_tau_main} \\
          |r(t)|    &\leq \epsilon_{t} \, t.    \label{eq:bound_r_t_main}
          \end{align}
\end{enumerate}
Thus $\epsilon_t$ furnishes a uniform control on the remainder throughout the entire interval $[0,t]$, and $\epsilon_t \to 0$ as $t \to 0^+$.

\vspace{0.5em}
\noindent
\textbf{Step 7: Precise estimation of the integral and completion of the limit computation.}

Employing the expansion of $x$ about $0$ with the constant $a_0 = \dot{x}(0)$ introduced in Step 6, we decompose $x(t)-x(\tau)$ into a principal part and a remainder term:
\begin{align*}
x(t) - x(\tau)
&= \bigl[ x(0) + a_0 t + r(t) \bigr] - \bigl[ x(0) + a_0 \tau + r(\tau) \bigr] \\
&= a_0(t - \tau) + \bigl[ r(t) - r(\tau) \bigr].
\end{align*}

Now evaluate the integral:
\begin{align*}
\int_0^t \bigl[ x(t) - x(\tau) \bigr] \, d\tau
&= \int_0^t a_0(t - \tau) \, d\tau + \int_0^t \bigl[ r(t) - r(\tau) \bigr] \, d\tau \\
&= a_0 \int_0^t (t - \tau) \, d\tau + \int_0^t r(t) \, d\tau - \int_0^t r(\tau) \, d\tau.
\end{align*}

The first term was computed exactly in Step 1: $\int_0^t (t-\tau)d\tau = t^{2}/2$, whence
\[
a_0 \int_0^t (t - \tau) \, d\tau = a_0 \cdot \frac{t^{2}}{2}.
\]

For the second term (the remainder part) we estimate each component separately. First,
\[
\Bigl| \int_0^t r(t) \, d\tau \Bigr|
= \Bigl|\, r(t) \int_0^t 1 \, d\tau \Bigr|
= |r(t)| \cdot \int_0^t 1 \, d\tau
= |r(t)| \cdot t.
\]
By the bound \eqref{eq:bound_r_t_main}, $|r(t)| \leq \epsilon_{t} t$, consequently
\[
\Bigl| \int_0^t r(t) \, d\tau \Bigr| \leq \epsilon_{t} t \cdot t = \epsilon_{t} t^{2}.
\]

Second,
\[
\Bigl| \int_0^t r(\tau) \, d\tau \Bigr| \leq \int_0^t |r(\tau)| \, d\tau.
\]
Using \eqref{eq:bound_r_tau_main}, $|r(\tau)| \leq \epsilon_{t} \tau$ for all $0 \leq \tau \leq t$, we obtain
\[
\int_0^t |r(\tau)| \, d\tau \leq \int_0^t \epsilon_{t} \tau \, d\tau = \epsilon_{t} \cdot \frac{t^{2}}{2}.
\]

Thus the total remainder satisfies
\begin{align*}
\Bigl| \int_0^t \bigl[ r(t) - r(\tau) \bigr] \, d\tau \Bigr|
&\leq \Bigl| \int_0^t r(t) \, d\tau \Bigr| + \Bigl| \int_0^t r(\tau) \, d\tau \Bigr| \\
&\leq \epsilon_{t} t^{2} + \epsilon_{t} \cdot \frac{t^{2}}{2} = \frac{3}{2} \epsilon_{t} t^{2}.
\end{align*}

Define $R(t) := \displaystyle\int_0^t \bigl[ r(t) - r(\tau) \bigr]  \, d\tau$, which constitutes precisely the remainder part of the integral. From the estimate above we have $|R(t)| \leq \frac{3}{2} \epsilon_{t} t^{2}$, and therefore
\begin{equation}
\label{eq:integral_with_error_main}
\int_0^t \bigl[ x(t) - x(\tau) \bigr]  \, d\tau = a_0 \cdot \frac{t^{2}}{2} + R(t).
\end{equation}

Substituting \eqref{eq:integral_with_error_main} into the simplified expression \eqref{eq:V11_simplified} yields
\begin{align*}
\mathscr{V}_{1,1}[x](t)
&= \frac{2}{t^{2}} \Bigl( a_0 \cdot \frac{t^{2}}{2} + R(t) \Bigr) \\
&= a_0 + \frac{2R(t)}{t^{2}}.
\end{align*}

From the estimate $|R(t)| \leq \frac{3}{2} \epsilon_{t} t^{2}$ we deduce
\[
\Bigl| \frac{2R(t)}{t^{2}} \Bigr| \leq \frac{2 \cdot \frac{3}{2} \epsilon_{t} t^{2}}{t^{2}} = 3 \epsilon_{t}.
\]

\vspace{0.5em}
\noindent
\textbf{Step 8: Completion of the limit argument.}

Since $\bigl| \frac{2R(t)}{t^{2}} \bigr| \leq 3\epsilon_{t}$ and, by Proposition~\ref{prop:epsilon_t_explicit_app}, $\epsilon_{t} \to 0$ as $t \to 0^{+}$, the squeeze theorem (if $|f(t)| \leq g(t)$ and $\lim_{t\to a} g(t) = 0$, then $\lim_{t\to a} f(t) = 0$) yields
\[
\lim_{t \to 0^{+}} \frac{2R(t)}{t^{2}} = 0.
\]

Assembling the preceding results,
\begin{align*}
\lim_{t \to 0^{+}} \mathscr{V}_{1,1}[x](t) 
&= \lim_{t \to 0^{+}} \Bigl[ a_0 + \frac{2R(t)}{t^{2}} \Bigr] \\
&= a_0 + \lim_{t \to 0^{+}} \frac{2R(t)}{t^{2}} \\
&= a_0 + 0 = a_0 = \dot{x}(0).
\end{align*}

Thus we have established rigorously
\[
\lim_{t \to 0^{+}} \mathscr{V}_{1,1}[x](t) = \dot{x}(0),
\]
demonstrating that in the uniform‑memory case $\alpha=\beta\equiv1$, the memory‑weighted velocity operator recovers the classical derivative at the origin.
\end{proof}

\begin{corollary}[Asymptotic location of the mean‑value point]
\label{cor:asymptotic_position}
Under the hypotheses of Theorem~\ref{thm:uniform_memory_limit}, assume additionally that 
$\dot{x}(0) \neq 0$. Let $\xi(t) \in [0, t]$ denote the point furnished by the first 
mean‑value theorem for integrals in the proof of that theorem (see equation \eqref{eq:mean_value_numerator}); that is, $\xi(t)$ satisfies
\[
\int_{0}^{t} \bigl[ x(t) - x(\tau) \bigr] \, d\tau = \bigl[ x(t) - x(\xi(t)) \bigr] \cdot t,
\qquad t \in (0, T].
\]
Then $\xi(t)$ exhibits the following precise asymptotic behavior:
\begin{equation}\label{eq:asymptotic_midpoint}
\lim_{t \to 0^{+}} \frac{\xi(t)}{t} = \frac{1}{2}.
\end{equation}
In words, the mean‑value point $\xi(t)$ asymptotically approaches the midpoint of the 
interval $[0, t]$ as $t \to 0^{+}$.
\end{corollary}

\begin{proof}
The proof follows directly from the relations established in the proof of Theorem~\ref{thm:uniform_memory_limit} together with the limit result of the theorem itself.

\vspace{0.5em}
\noindent
\textbf{Step 1: Recalling the essential identities.}

From Step 2 of the theorem's proof (see equation \eqref{eq:mean_value_numerator}), 
there exists $\xi(t) \in [0, t]$ such that
\begin{equation}
\label{eq:cor_mean_value_revisit}
\int_{0}^{t} \bigl[ x(t) - x(\tau) \bigr] \, d\tau = \bigl[ x(t) - x(\xi(t)) \bigr] \cdot t.
\end{equation}
Introduce the normalized position
\[
\theta(t) := \frac{\xi(t)}{t} \in [0, 1].
\]

Step 4 of the theorem's proof (equation \eqref{eq:lagrange_general}) supplies a point 
$\zeta(t) \in [0, t]$ for which
\begin{equation}
\label{eq:cor_lagrange_revisit}
x(t) - x(\xi(t)) = \dot{x}(\zeta(t)) \cdot (t - \xi(t)) 
                  = \dot{x}(\zeta(t)) \cdot t\bigl(1 - \theta(t)\bigr).
\end{equation}

Finally, from Step 5 (equation \eqref{eq:V11_theta_form}) we have the compact representation
\begin{equation}
\label{eq:cor_V11_form_revisit}
\mathscr{V}_{1,1}[x](t) = 2 \dot{x}(\zeta(t)) \cdot \bigl(1 - \theta(t)\bigr).
\end{equation}

\vspace{0.5em}
\noindent
\textbf{Step 2: Applying the limit from Theorem~\ref{thm:uniform_memory_limit}.}

Theorem~\ref{thm:uniform_memory_limit} asserts that
\begin{equation}
\label{eq:cor_limit_V11}
\lim_{t \to 0^{+}} \mathscr{V}_{1,1}[x](t) = \dot{x}(0).
\end{equation}
Moreover, since $x \in \mathcal{C}^{1}(I)$, its derivative $\dot{x}$ is continuous on the compact interval $I$, and in particular at the point $0$.

Observe that $\zeta(t) \in [0, t]$ implies $0 \leq \zeta(t) \leq t$. Consequently, $\zeta(t) \to 0$ as $t \to 0^{+}$. 

Because $\dot{x}$ is continuous at $0$, it preserves limits: for any sequence (or function) converging to $0$, the image under $\dot{x}$ converges to $\dot{x}(0)$. Applying this observation to the function $\zeta(t) \to 0$, we obtain
\begin{equation}
\label{eq:cor_limit_zeta}
\lim_{t \to 0^{+}} \dot{x}(\zeta(t)) = \dot{x}(0).
\end{equation}

\vspace{0.5em}
\noindent
\textbf{Step 3: Determining the limit of $\theta(t)$.}

From the compact representation \eqref{eq:cor_V11_form_revisit} we have for all $t>0$ that
$\mathscr{V}_{1,1}[x](t) = 2\dot{x}(\zeta(t)) \cdot \bigl(1 - \theta(t)\bigr)$.

Since $\dot{x}(0) \neq 0$ by hypothesis and $\dot{x}$ is continuous on $I$, there exists $\rho > 0$ 
such that $\dot{x}(s) \neq 0$ for all $s \in [0,\rho)$. Because $\zeta(t) \in [0,t]$ and 
$\zeta(t) \to 0$ as $t \to 0^{+}$, for all sufficiently small $t>0$ we have 
$\zeta(t) \in [0,\rho)$ and consequently $\dot{x}(\zeta(t)) \neq 0$. For such $t$ we can safely solve 
for $\theta(t)$:
\begin{equation}
\label{eq:theta_explicit}
\theta(t) = 1 - \frac{\mathscr{V}_{1,1}[x](t)}{2\dot{x}(\zeta(t))}.
\end{equation}

Now consider the limit $t \to 0^{+}$. By Theorem~\ref{thm:uniform_memory_limit} and the continuity of $\dot{x}$ at $0$,
\[
\lim_{t \to 0^{+}} \mathscr{V}_{1,1}[x](t) = \dot{x}(0), \qquad
\lim_{t \to 0^{+}} \dot{x}(\zeta(t)) = \dot{x}(0) \neq 0.
\]

Since the denominator limit $2\dot{x}(0)$ is non‑zero, applying the quotient rule for limits to \eqref{eq:theta_explicit} yields
\[
\lim_{t \to 0^{+}} \theta(t) = 
1 - \frac{\displaystyle\lim_{t \to 0^{+}} \mathscr{V}_{1,1}[x](t)}
          {\displaystyle 2\lim_{t \to 0^{+}} \dot{x}(\zeta(t))}
= 1 - \frac{\dot{x}(0)}{2\dot{x}(0)} = \frac{1}{2}.
\]

\vspace{0.5em}
\noindent
\textbf{Step 4: Conclusion for $\xi(t)$.}

Recalling that $\theta(t) = \xi(t)/t$, we immediately deduce
\[
\lim_{t \to 0^{+}} \frac{\xi(t)}{t} = \lim_{t \to 0^{+}} \theta(t) = \frac{1}{2}.
\]

\vspace{0.5em}
\noindent
\textbf{Step 5: Interpretation and completion of the proof.}

Thus, as the time interval $[0, t]$ contracts to zero, the relative position of the 
mean‑value point $\xi(t)$ within the interval converges to the midpoint. 
This completes the proof of Corollary~\ref{cor:asymptotic_position}.
\end{proof}

\begin{remark}[The uniform‑memory case as a consistency check]
\label{rem:uniform_memory_consistency}

The results of Theorem~\ref{thm:uniform_memory_limit} and Corollary~\ref{cor:asymptotic_position} together illustrate how the memory‑weighted velocity operator $\mathscr{V}_{\alpha,\beta}$ connects to classical differential calculus in the simplest non‑trivial setting.

\textbf{Recovery of the classical derivative.}
When both memory exponents are set to unity, the operator reduces to the purely averaged form \eqref{eq:V11_form}, where past information is weighted uniformly over $[0,t]$. Theorem~\ref{thm:uniform_memory_limit} establishes that despite this non‑local averaging, $\mathscr{V}_{1,1}[x](t)$ converges to the ordinary derivative $\dot{x}(0)$ as $t\to0^{+}$. The detailed analysis presented here—based on explicit integral estimates and the systematic error‑control functions constructed in Appendix~\ref{app:error_control}—provides a self‑contained and conceptually transparent derivation that aligns naturally with the analytical framework developed throughout this work. This approach not only confirms the consistency of the definition $\mathscr{V}_{1,1}[x](0)=\dot{x}(0)$ with the limiting behavior of the integral expression but also deepens the understanding of how the averaging process interacts with differentiability.

\textbf{Geometric interpretation of the averaging process.}
Under the condition $\dot{x}(0)\neq0$, Corollary~\ref{cor:asymptotic_position} reveals that the mean‑value point $\xi(t)$ furnished by the integral mean‑value theorem satisfies $\xi(t)/t\to 1/2$. This asymptotic midpoint property indicates that in the uniform‑memory regime, the historical information is balanced symmetrically about the center of the interval as the observation window shrinks. The averaging inherent in $\mathscr{V}_{1,1}$ thus distributes the influence of the past in an even‑handed manner when no preferential weighting is present.

\textbf{Transition from memory-dependent to local behavior.}
The limit $t\to 0^{+}$ in Theorem~\ref{thm:uniform_memory_limit} represents a gradual transition from a non-local, memory-inclusive description to a purely local, memoryless one. The convergence $\mathscr{V}_{1,1}[x](t) \to \dot{x}(0)$ demonstrates that as the time interval over which history is considered contracts to zero, the memory-weighted rate of change smoothly approaches the classical instantaneous derivative. This continuity between the memory-dependent formulation and the local limit reflects an inherent consistency in the mathematical framework: the incorporation of historical information does not introduce discontinuities or artifacts but rather extends the classical concept in a manner that preserves the expected local behavior in the appropriate limit.

\textbf{Significance for the general framework.}
Together, these results offer a fundamental consistency verification for the construction of $\mathscr{V}_{\alpha,\beta}$. They demonstrate that even in this elementary special case—where memory is made uniform—the operator preserves a transparent link with ordinary differentiation, while its integral structure already exhibits the smoothing and balancing characteristics typical of memory‑weighted rates of change. The uniform‑memory case therefore serves as a valuable reference point for understanding the more general situation where the memory exponents $\alpha(t)$ and $\beta(t)$ vary with time. From the perspective of operator theory, this case also provides a concrete example of how a family of integral operators can interpolate between nonlocal and local descriptions—a theme of interest in the analysis of parameter-dependent operator families arising in nonlinear dynamics.
\end{remark}

\subsection{Weighted Estimates and Mapping Properties: Boundedness Analysis of $\mathscr{V}_{\alpha,\beta}$} 
\label{subsec:weighted_estimates_boundedness}

Having laid the groundwork with the foundational properties established in the preceding sections, we now turn to examine the mapping properties of the memory-weighted velocity operator between standard function spaces. A central question is whether $\mathscr{V}_{\alpha,\beta}$ defines a bounded linear operator from $\mathcal{C}^{1}(I)$ into $\mathcal{C}(I)$. As we shall see, the boundedness property hinges essentially on the relative magnitude of the two memory exponents $\alpha(t)$ and $\beta(t)$. To render this relationship explicit, we develop a \emph{weighted pointwise estimate} that delineates how the exponent difference $\beta(t)-\alpha(t)$ modulates the magnitude of $\mathscr{V}_{\alpha,\beta}[x](t)$. The precise formulation follows.

\begin{theorem}[Exponent-difference modulated pointwise bound for $\mathscr{V}_{\alpha,\beta}$]
\label{thm:weighted_boundedness}
Let $\alpha, \beta \in \mathcal{C}(I)$ satisfy $0 < m \le \alpha(t), \beta(t) \le M \le 1$ for every $t \in I$. 
Define the extremal Gamma‑function values
\[
\Gamma_{\mathrm{low}} := \min_{z \in [m, M]} \Gamma(z) > 0, \qquad 
\Gamma_{\mathrm{up}} := \max_{z \in [m, M]} \Gamma(z),
\]
and introduce the constant
\begin{equation}\label{eq:weighted_constant_explicit}
C_W := \frac{(M+1)\Gamma_{\mathrm{up}}}{(m+1)\Gamma_{\mathrm{low}}} \ge 1.
\end{equation}
(The inequality $C_W \ge 1$ follows immediately from $M \ge m$ and $\Gamma_{\mathrm{up}} \ge \Gamma_{\mathrm{low}}$.)

\textbf{(i) For $t \in (0,T]$:} For any $x \in \mathcal{C}^{1}(I)$, the memory-weighted velocity operator obeys the weighted estimate
\begin{equation}\label{eq:weighted_inequality}
\min\!\bigl\{1, t^{\beta(t)-\alpha(t)}\bigr\} \cdot 
\bigl|\mathscr{V}_{\alpha,\beta}[x](t)\bigr| \le C_W \|x\|_{\mathcal{C}^{1}}.
\end{equation}

\textbf{(ii) At the initial instant $t=0$:} For any $x \in \mathcal{C}^{1}(I)$,
\[
\bigl|\mathscr{V}_{\alpha,\beta}[x](0)\bigr| \le C_W \|x\|_{\mathcal{C}^{1}}.
\]
\end{theorem}

\begin{proof}
We provide a systematic proof organized into six logical steps. The argument combines the explicit representation of $\mathscr{V}_{\alpha,\beta}$ with elementary properties of the Gamma function, Lipschitz continuity of $\mathcal{C}^{1}$ functions, and a careful case analysis of the weighting factor modulated by the exponent difference.

\textbf{Step 1: Preliminary decomposition and coefficient bounds.}

For $t \in (0,T]$, recall from Corollary~\ref{cor:simplified_form} that
\[
\mathscr{V}_{\alpha,\beta}[x](t) = \frac{(\beta(t)+1)\Gamma(\beta(t))}{\Gamma(\alpha(t))} \cdot 
\frac{1}{t^{\beta(t)+1}} \int_{0}^{t} (t-\tau)^{\alpha(t)-1} \bigl[x(t)-x(\tau)\bigr] d\tau.
\]

Introduce the coefficient function
\[
\Theta(t) := \frac{(\beta(t)+1)\Gamma(\beta(t))}{\Gamma(\alpha(t))}.
\]

Since the Gamma function is continuous and strictly positive on $(0,\infty)$, and both $\alpha(t)$ and $\beta(t)$ belong to the compact interval $[m, M]$, we obtain the uniform bounds
\[
\Gamma_{\mathrm{low}} \le \Gamma(\alpha(t)), \Gamma(\beta(t)) \le \Gamma_{\mathrm{up}} \quad \text{for all } t \in I,
\]
where $\Gamma_{\mathrm{low}} := \min_{z \in [m, M]} \Gamma(z) > 0$ and $\Gamma_{\mathrm{up}} := \max_{z \in [m, M]} \Gamma(z)$.

Consequently,
\begin{equation}\label{eq:Theta_bound}
|\Theta(t)| \le \frac{(M+1)\Gamma_{\mathrm{up}}}{\Gamma_{\mathrm{low}}} =: B_{\Theta} \quad \text{uniformly on } I.
\end{equation}
Observe that $B_{\Theta}$ differs from the constant $C_W$ defined in \eqref{eq:weighted_constant_explicit} precisely by the factor $1/(m+1)$; this separation will streamline the subsequent estimates.

\textbf{Step 2: Lipschitz estimate and integral control.}

Since $x \in \mathcal{C}^{1}(I)$, the function is Lipschitz continuous on the compact interval $I$. 
Let $L_x := \|\dot{x}\|_{\infty}$; this serves as a Lipschitz constant for $x$ on $I$, 
since for any $t, \tau \in I$,
\[
|x(t) - x(\tau)| = \left| \int_{\tau}^{t} \dot{x}(s) ds \right| \leq \|\dot{x}\|_{\infty} |t-\tau| = L_x |t-\tau|.
\]

Applying this Lipschitz estimate to the integral term yields
\begin{align*}
\biggl| \int_{0}^{t} (t-\tau)^{\alpha(t)-1} [x(t)-x(\tau)] d\tau \biggr|
&\le L_x \int_{0}^{t} (t-\tau)^{\alpha(t)-1} |t-\tau| d\tau \\
&= L_x \int_{0}^{t} (t-\tau)^{\alpha(t)} d\tau.
\end{align*}

A direct computation via the change of variable $u = t-\tau$ gives
\[
\int_{0}^{t} (t-\tau)^{\alpha(t)} d\tau = \int_{0}^{t} u^{\alpha(t)} du = \frac{t^{\alpha(t)+1}}{\alpha(t)+1}.
\]

Moreover, since $\alpha(t) \ge m > 0$, we have $\alpha(t)+1 \ge m+1$, and thus
\[
\frac{1}{\alpha(t)+1} \le \frac{1}{m+1}.
\]

Combining these estimates produces the bound
\begin{equation}\label{eq:integral_bound}
\biggl| \int_{0}^{t} (t-\tau)^{\alpha(t)-1} [x(t)-x(\tau)] d\tau \biggr| \le 
\frac{L_x}{m+1} \, t^{\alpha(t)+1}.
\end{equation}

\textbf{Step 3: Derivation of an elementary pointwise estimate.}

From \eqref{eq:Theta_bound} and \eqref{eq:integral_bound}, we obtain
\begin{align*}
\bigl|\mathscr{V}_{\alpha,\beta}[x](t)\bigr| 
&\le |\Theta(t)| \cdot \frac{1}{t^{\beta(t)+1}} \cdot \biggl| \int_{0}^{t} (t-\tau)^{\alpha(t)-1} [x(t)-x(\tau)] d\tau \biggr| \\
&\le B_{\Theta} \cdot \frac{1}{t^{\beta(t)+1}} \cdot \frac{L_x}{m+1} \, t^{\alpha(t)+1} \\
&= \frac{B_{\Theta} L_x}{m+1} \cdot t^{\alpha(t)-\beta(t)} \\
&= \frac{(M+1)\Gamma_{\mathrm{up}}}{\Gamma_{\mathrm{low}}} \cdot \frac{L_x}{m+1} \cdot t^{\alpha(t)-\beta(t)} \quad \text{(by definition of } B_{\Theta} \text{ in \eqref{eq:Theta_bound})}.
\end{align*}

Hence, recalling that $L_x = \|\dot{x}\|_{\infty}$,
\begin{equation}\label{eq:elementary_pointwise}
\bigl|\mathscr{V}_{\alpha,\beta}[x](t)\bigr| \le \frac{(M+1)\Gamma_{\mathrm{up}}}{(m+1)\Gamma_{\mathrm{low}}} \, \|\dot{x}\|_{\infty} \cdot t^{\alpha(t)-\beta(t)}.
\end{equation}

\textbf{Step 4: Analysis of the weighting factor modulated by exponent difference.}

Define $\Delta(t) := \beta(t) - \alpha(t)$. Consider the product
\[
\Pi(t) := \min\{1, t^{\Delta(t)}\} \cdot t^{\alpha(t)-\beta(t)}.
\]

We demonstrate that $\Pi(t) \le 1$ for all $t \in (0,T]$. Distinguish two principal cases:

\emph{Case 1: $\Delta(t) \ge 0$ (i.e., $\beta(t) \ge \alpha(t)$).}

- \textbf{Subcase 1.1: $0 < t \le 1$.} Since $t \le 1$ and $\Delta(t) \ge 0$, we have $t^{\Delta(t)} \le 1$, whence $\min\{1, t^{\Delta(t)}\} = t^{\Delta(t)}$. Consequently,
  \[
  \Pi(t) = t^{\Delta(t)} \cdot t^{-\Delta(t)} = 1.
  \]
  
- \textbf{Subcase 1.2: $1 < t \le T$.} Here $t > 1$ and $\Delta(t) \ge 0$ imply $t^{\Delta(t)} \ge 1$, so $\min\{1, t^{\Delta(t)}\} = 1$. Moreover, for $t > 1$, the function $s \mapsto t^{s}$ is strictly increasing in the exponent $s$. Since $-\Delta(t) \le 0$, we have $t^{-\Delta(t)} \le t^{0} = 1$. Therefore
  \[
  \Pi(t) = 1 \cdot t^{-\Delta(t)} \le 1.
  \]

\emph{Case 2: $\Delta(t) < 0$ (i.e., $\beta(t) < \alpha(t)$).}

Set $\eta(t) := -\Delta(t) = \alpha(t) - \beta(t) > 0$. Then $t^{\Delta(t)} = t^{-\eta(t)}$.

- \textbf{Subcase 2.1: $0 < t \le 1$.} For $0 < t \le 1$, we have $t^{-\eta(t)} \ge 1$, so $\min\{1, t^{\Delta(t)}\} = 1$. Furthermore, since $t \le 1$ and $\eta(t) > 0$, the inequality $t^{\eta(t)} \le 1$ holds. Hence
  \[
  \Pi(t) = 1 \cdot t^{\eta(t)} \le 1.
  \]
  
- \textbf{Subcase 2.2: $1 < t \le T$.} When $t > 1$ and $\eta(t) > 0$, the strict inequality $t^{-\eta(t)} < 1$ holds, whence $\min\{1, t^{\Delta(t)}\} = t^{-\eta(t)}$. Thus
  \[
  \Pi(t) = t^{-\eta(t)} \cdot t^{\eta(t)} = 1.
  \]

In every scenario, $\Pi(t) \le 1$. Summarizing,
\begin{equation}\label{eq:Pi_inequality}
\min\{1, t^{\beta(t)-\alpha(t)}\} \cdot t^{\alpha(t)-\beta(t)} \le 1 \quad \text{for all } t \in (0,T].
\end{equation}

\textbf{Step 5: Completion of the weighted estimate.}

Since $t>0$ and both $1$ and $t^{\beta(t)-\alpha(t)}$ are positive, their minimum $\min\{1, t^{\beta(t)-\alpha(t)}\}$ is also positive. Multiplying both sides of \eqref{eq:elementary_pointwise} by this positive factor and employing \eqref{eq:Pi_inequality}, we obtain
\begin{align*}
\min\{1, t^{\beta(t)-\alpha(t)}\} \cdot \bigl|\mathscr{V}_{\alpha,\beta}[x](t)\bigr|
&\le \frac{(M+1)\Gamma_{\mathrm{up}}}{(m+1)\Gamma_{\mathrm{low}}} \, \|\dot{x}\|_{\infty} \cdot 
\min\{1, t^{\beta(t)-\alpha(t)}\} \cdot t^{\alpha(t)-\beta(t)} \\
&\le \frac{(M+1)\Gamma_{\mathrm{up}}}{(m+1)\Gamma_{\mathrm{low}}} \, \|\dot{x}\|_{\infty} \quad \text{(by \eqref{eq:Pi_inequality})}.
\end{align*}

At this stage we have established a bound involving the sup-norm of the derivative $\|\dot{x}\|_{\infty}$. To express the estimate in terms of the $\mathcal{C}^{1}$-norm of $x$, we invoke the relation $\|\dot{x}\|_{\infty} \le \|x\|_{\mathcal{C}^{1}}$, which follows directly from the definition of the $\mathcal{C}^{1}$-norm in \eqref{eq:C1_norm}. Consequently,
\[
\min\{1, t^{\beta(t)-\alpha(t)}\} \cdot \bigl|\mathscr{V}_{\alpha,\beta}[x](t)\bigr| \le 
\frac{(M+1)\Gamma_{\mathrm{up}}}{(m+1)\Gamma_{\mathrm{low}}} \|x\|_{\mathcal{C}^{1}}.
\]

The expression $\dfrac{(M+1)\Gamma_{\mathrm{up}}}{(m+1)\Gamma_{\mathrm{low}}}$ appearing in the estimate is exactly the constant $C_W$ introduced in the theorem statement. Hence we have established
\[
\min\{1, t^{\beta(t)-\alpha(t)}\} \cdot \bigl|\mathscr{V}_{\alpha,\beta}[x](t)\bigr| \le C_W \|x\|_{\mathcal{C}^{1}},
\]
which is precisely the weighted inequality \eqref{eq:weighted_inequality}. This completes the proof of part (i) of the theorem.

\textbf{Step 6: Verification at the initial time.}

According to Definition~\ref{def:velocity_operator}, $\mathscr{V}_{\alpha,\beta}[x](0) = \dot{x}(0)$. Therefore
\[
\bigl|\mathscr{V}_{\alpha,\beta}[x](0)\bigr| = |\dot{x}(0)| \le \|\dot{x}\|_{\infty} \le \|x\|_{\mathcal{C}^{1}},
\]
where the last inequality again follows from the definition of the $\mathcal{C}^{1}$-norm.

Since $C_W \ge 1$ by definition \eqref{eq:weighted_constant_explicit}, we have $\|x\|_{\mathcal{C}^{1}} \le C_W \|x\|_{\mathcal{C}^{1}}$. Combining this with the preceding estimate yields
\[
\bigl|\mathscr{V}_{\alpha,\beta}[x](0)\bigr| \le C_W \|x\|_{\mathcal{C}^{1}}.
\]

Thus part (ii) of the theorem is established. Together with part (i) proved in Steps 1--5, the proof is complete.
\end{proof}

The weighted estimate established in Theorem~\ref{thm:weighted_boundedness} furnishes a versatile tool for analyzing the mapping properties of $\mathscr{V}_{\alpha,\beta}$. To obtain more refined pointwise bounds that explicitly reflect the extremal behavior of the exponent difference $\beta(t)-\alpha(t)$, we now examine several special configurations based on the minimum and maximum values of this difference over the interval $I$.

\begin{corollary}[Fine estimates based on extremal values of the exponent difference]\label{cor:extremal_estimates}
Under the hypotheses of Theorem~\ref{thm:weighted_boundedness}, define the extremal values of the exponent difference:
\begin{equation}\label{eq:extremal_definitions}
\Delta_{\min} := \min_{t \in I} [\beta(t) - \alpha(t)], \qquad
\Delta_{\max} := \max_{t \in I} [\beta(t) - \alpha(t)].
\end{equation}
Since $\alpha, \beta \in \mathcal{C}(I)$ and $I = [0,T]$ is compact, these extrema are attained. From the bounds $m \le \alpha(t), \beta(t) \le M$, we have
\begin{equation}\label{eq:extremal_bounds}
m - M \le \Delta_{\min} \le \beta(t) - \alpha(t) \le \Delta_{\max} \le M - m \quad \text{for all } t \in I.
\end{equation}

\textbf{(i) For $t \in (0,T]$:} The following detailed pointwise estimates hold according to the signs of $\Delta_{\min}$ and $\Delta_{\max}$.

\begin{enumerate}
    \item \textbf{Case A: $\Delta_{\max} \le 0$ (i.e., $\beta(t) \le \alpha(t)$ for all $t \in I$).}
    
    In this situation $\beta(t) - \alpha(t) \le 0$ throughout $I$.
    
    \begin{itemize}
        \item \textbf{Subcase A1: $t \in (0, 1]$.}
        Since $\beta(t) - \alpha(t) \le 0$ and $0 < t \le 1$, we have
        \[
        t^{\beta(t)-\alpha(t)} \ge t^0 = 1.
        \]
        Hence $\min\{1, t^{\beta(t)-\alpha(t)}\} = 1$. Substituting this into the weighted inequality \eqref{eq:weighted_inequality} yields immediately
        \[
        \bigl|\mathscr{V}_{\alpha,\beta}[x](t)\bigr| \le C_W \|x\|_{\mathcal{C}^{1}} \quad \text{for all } t \in (0, 1].
        \]
        
        \item \textbf{Subcase A2: $t \in (1, T]$.}
        Since $\beta(t) - \alpha(t) \le 0$ and $t > 1$, we have
        \[
        t^{\beta(t)-\alpha(t)} \le t^0 = 1,
        \]
        so $\min\{1, t^{\beta(t)-\alpha(t)}\} = t^{\beta(t)-\alpha(t)}$.
        
        Applying the weighted inequality \eqref{eq:weighted_inequality} with this specific value of the minimum gives
        \[
        t^{\beta(t)-\alpha(t)} \bigl|\mathscr{V}_{\alpha,\beta}[x](t)\bigr| \le C_W \|x\|_{\mathcal{C}^{1}}.
        \]
        Equivalently,
        \[
        \bigl|\mathscr{V}_{\alpha,\beta}[x](t)\bigr| \le C_W \|x\|_{\mathcal{C}^{1}} \cdot t^{\alpha(t)-\beta(t)}.
        \]
        
        Define $\eta(t) := \alpha(t) - \beta(t) \ge 0$. Since $t > 1$ and $\eta(t) \le M - m$,
        \[
        t^{\eta(t)} \le t^{M-m} \le T^{M-m}.
        \]
        Consequently,
        \[
        \bigl|\mathscr{V}_{\alpha,\beta}[x](t)\bigr| \le C_W T^{M-m} \|x\|_{\mathcal{C}^{1}} \quad \text{for all } t \in (1, T].
        \]
    \end{itemize}
    
    Combining Subcases A1 and A2, we conclude that when $\Delta_{\max} \le 0$,
\[
\bigl|\mathscr{V}_{\alpha,\beta}[x](t)\bigr| \le C_{\text{caseA}} \|x\|_{\mathcal{C}^{1}} \quad \text{for all } t \in (0,T],
\]
with the explicit constant
\[
C_{\text{caseA}} := \max\bigl\{ C_W,\; C_W T^{M-m} \bigr\} = C_W \cdot \max\bigl\{ 1,\; T^{M-m} \bigr\}.
\]
This yields the pointwise estimate for configurations where $\beta(t) \le \alpha(t)$ holds for all $t \in (0,T]$ (i.e., $\Delta_{\max} \le 0$).
    
    \item \textbf{Case B: $\Delta_{\min} \ge 0$ (i.e., $\beta(t) \ge \alpha(t)$ for all $t \in I$).}
    
    Here $\beta(t) - \alpha(t) \ge 0$ on $I$.
    
    \begin{itemize}
        \item \textbf{Subcase B1: $t \in (0, 1]$.}
        Since $\beta(t) - \alpha(t) \ge 0$ and $0 < t \le 1$, we have
        \[
        t^{\beta(t)-\alpha(t)} \le t^0 = 1,
        \]
        hence $\min\{1, t^{\beta(t)-\alpha(t)}\} = t^{\beta(t)-\alpha(t)}$. The weighted inequality \eqref{eq:weighted_inequality} then gives
        \[
        t^{\beta(t)-\alpha(t)} \bigl|\mathscr{V}_{\alpha,\beta}[x](t)\bigr| \le C_W \|x\|_{\mathcal{C}^{1}}.
        \]
        
        Since $\beta(t) - \alpha(t) \le \Delta_{\max}$ and for $t \in (0,1]$, $t^a$ is non-increasing in $a$, we have
        \[
        t^{\beta(t)-\alpha(t)} \ge t^{\Delta_{\max}}.
        \]
        Therefore
        \[
        t^{\Delta_{\max}} \bigl|\mathscr{V}_{\alpha,\beta}[x](t)\bigr| \le C_W \|x\|_{\mathcal{C}^{1}} \quad \Rightarrow \quad
        \bigl|\mathscr{V}_{\alpha,\beta}[x](t)\bigr| \le C_W t^{-\Delta_{\max}} \|x\|_{\mathcal{C}^{1}}.
        \]
        
        \item \textbf{Subcase B2: $t \in (1, T]$.}
        Since $\beta(t) - \alpha(t) \ge 0$ and $t > 1$, we have
        \[
        t^{\beta(t)-\alpha(t)} \ge t^0 = 1,
        \]
        so $\min\{1, t^{\beta(t)-\alpha(t)}\} = 1$. The weighted inequality \eqref{eq:weighted_inequality} then yields directly
        \[
        \bigl|\mathscr{V}_{\alpha,\beta}[x](t)\bigr| \le C_W \|x\|_{\mathcal{C}^{1}}.
        \]
    \end{itemize}
    
    Summarizing Subcases B1 and B2, when $\Delta_{\min} \ge 0$,
    \[
    \bigl|\mathscr{V}_{\alpha,\beta}[x](t)\bigr| \le 
    \begin{cases}
    C_W t^{-\Delta_{\max}} \|x\|_{\mathcal{C}^{1}}, & 0 < t \le 1, \\
    C_W \|x\|_{\mathcal{C}^{1}}, & 1 < t \le T.
    \end{cases}
    \]
    
        \item \textbf{Case C: $\Delta_{\min} < 0 < \Delta_{\max}$ (mixed case).}
    
    In this situation $\beta(t) - \alpha(t)$ can take both negative and positive values, but always satisfies $\Delta_{\min} \le \beta(t) - \alpha(t) \le \Delta_{\max}$.
    
    \begin{itemize}
        \item \textbf{Subcase C1: $t \in (0, 1]$.}
        For $t \in (0,1]$, since $t^a$ is non-increasing in $a$, we have
        \[
        t^{\Delta_{\max}} \le t^{\beta(t)-\alpha(t)} \le t^{\Delta_{\min}}.
        \]
        Observe that $t^{\Delta_{\max}} \le 1$ (since $\Delta_{\max} > 0$ and $t \le 1$) and $t^{\Delta_{\min}} \ge 1$ (since $\Delta_{\min} < 0$ and $t \le 1$).
        
        From the weighted inequality \eqref{eq:weighted_inequality},
        \[
        \min\{1, t^{\beta(t)-\alpha(t)}\} \cdot \bigl|\mathscr{V}_{\alpha,\beta}[x](t)\bigr| \le C_W \|x\|_{\mathcal{C}^{1}}.
        \]
        
        We now establish a lower bound for the factor $\min\{1, t^{\beta(t)-\alpha(t)}\}$. Two possibilities occur:
        \begin{itemize}
            \item If $\beta(t) - \alpha(t) \ge 0$, then $\min\{1, t^{\beta(t)-\alpha(t)}\} = t^{\beta(t)-\alpha(t)} \ge t^{\Delta_{\max}}$ (since $t^{\beta(t)-\alpha(t)} \ge t^{\Delta_{\max}}$ when $t \le 1$).
            \item If $\beta(t) - \alpha(t) < 0$, then $\min\{1, t^{\beta(t)-\alpha(t)}\} = 1 \ge t^{\Delta_{\max}}$ (because $t^{\Delta_{\max}} \le 1$).
        \end{itemize}
        In either situation, $\min\{1, t^{\beta(t)-\alpha(t)}\} \ge t^{\Delta_{\max}}$. Consequently,
        \[
        t^{\Delta_{\max}} \bigl|\mathscr{V}_{\alpha,\beta}[x](t)\bigr| \le \min\{1, t^{\beta(t)-\alpha(t)}\} \cdot \bigl|\mathscr{V}_{\alpha,\beta}[x](t)\bigr| \le C_W \|x\|_{\mathcal{C}^{1}},
        \]
        which yields
        \[
        \bigl|\mathscr{V}_{\alpha,\beta}[x](t)\bigr| \le C_W t^{-\Delta_{\max}} \|x\|_{\mathcal{C}^{1}}.
        \]
        
        \item \textbf{Subcase C2: $t \in (1, T]$.}
        For $t > 1$, since $t^a$ is increasing in $a$, we have
        \[
        t^{\Delta_{\min}} \le t^{\beta(t)-\alpha(t)} \le t^{\Delta_{\max}}.
        \]
        Note that $t^{\Delta_{\min}} \le 1$ (since $\Delta_{\min} < 0$ and $t > 1$) and $t^{\Delta_{\max}} \ge 1$ (since $\Delta_{\max} > 0$ and $t > 1$).
        
        Again we bound the factor $\min\{1, t^{\beta(t)-\alpha(t)}\}$ from below:
        \begin{itemize}
            \item If $\beta(t) - \alpha(t) \le 0$, then $\min\{1, t^{\beta(t)-\alpha(t)}\} = t^{\beta(t)-\alpha(t)} \ge t^{\Delta_{\min}}$ (since $t^{\beta(t)-\alpha(t)} \ge t^{\Delta_{\min}}$ when $t > 1$).
            \item If $\beta(t) - \alpha(t) > 0$, then $\min\{1, t^{\beta(t)-\alpha(t)}\} = 1 \ge t^{\Delta_{\min}}$ (because $t^{\Delta_{\min}} \le 1$).
        \end{itemize}
        Thus $\min\{1, t^{\beta(t)-\alpha(t)}\} \ge t^{\Delta_{\min}}$ in all cases. Applying \eqref{eq:weighted_inequality} gives
        \[
        t^{\Delta_{\min}} \bigl|\mathscr{V}_{\alpha,\beta}[x](t)\bigr| \le \min\{1, t^{\beta(t)-\alpha(t)}\} \cdot \bigl|\mathscr{V}_{\alpha,\beta}[x](t)\bigr| \le C_W \|x\|_{\mathcal{C}^{1}},
        \]
        which yields
        \[
        \bigl|\mathscr{V}_{\alpha,\beta}[x](t)\bigr| \le C_W t^{-\Delta_{\min}} \|x\|_{\mathcal{C}^{1}}.
        \]
    \end{itemize}
    
    Combining Subcases C1 and C2, when $\Delta_{\min} < 0 < \Delta_{\max}$,
    \[
    \bigl|\mathscr{V}_{\alpha,\beta}[x](t)\bigr| \le 
    \begin{cases}
    C_W t^{-\Delta_{\max}} \|x\|_{\mathcal{C}^{1}}, & 0 < t \le 1, \\
    C_W t^{-\Delta_{\min}} \|x\|_{\mathcal{C}^{1}}, & 1 < t \le T.
    \end{cases}
    \]
\end{enumerate}

\textbf{(ii) At the initial time $t=0$:} As established in part (ii) of Theorem~\ref{thm:weighted_boundedness},
\[
\bigl|\mathscr{V}_{\alpha,\beta}[x](0)\bigr| \le C_W \|x\|_{\mathcal{C}^{1}}.
\]

In all cases, the constant $C_W$ is that defined in Theorem~\ref{thm:weighted_boundedness} (see \eqref{eq:weighted_constant_explicit}):
\[
C_W = \frac{(M+1)\Gamma_{\mathrm{up}}}{(m+1)\Gamma_{\mathrm{low}}}.
\]
\end{corollary}
\begin{remark}[Interpretation and significance of the weighted estimates]\label{rem:weighted_estimates_significance}
The results presented in this section illuminate several distinctive features of the memory-weighted velocity operator $\mathscr{V}_{\alpha,\beta}$, revealing how its analytical character is intimately connected to the interplay between the two memory exponents.

\textbf{Modulation by the exponent difference and physical interpretation.} 
Theorem~\ref{thm:weighted_boundedness} demonstrates that the pointwise magnitude of $\mathscr{V}_{\alpha,\beta}[x](t)$ is intrinsically governed by the difference $\beta(t)-\alpha(t)$. The weighting factor $\min\{1, t^{\beta(t)-\alpha(t)}\}$ serves as a natural compensator: precisely when $\beta(t) > \alpha(t)$—so that the elementary estimate $t^{\alpha(t)-\beta(t)}$ could become unbounded as $t \to 0^+$—the factor $\min\{1, t^{\beta(t)-\alpha(t)}\}$ counterbalances this potential growth, thereby ensuring a uniform bound. From a physical perspective, this compensation mechanism reflects how systems with stronger time-scale memory (larger $\beta$) relative to state-increment memory (smaller $\alpha$) require additional regularization near the initial time to maintain bounded rates of change—a feature relevant to modeling processes where short-time behavior is particularly sensitive to historical weighting.

\textbf{Adaptive control through extremal values.} 
Corollary~\ref{cor:extremal_estimates} further refines this understanding by showing how the global extremal values $\Delta_{\min}$ and $\Delta_{\max}$ of the exponent difference translate into explicit pointwise bounds. The casewise analysis illustrates that the operator's behavior bifurcates according to whether $\beta(t)$ is consistently larger than, consistently smaller than, or variably related to $\alpha(t)$ over the interval. These distinct regimes correspond to different physical scenarios: when time-scale memory dominates ($\beta > \alpha$), when state-increment memory dominates ($\beta < \alpha$), or when the relationship varies, as might occur in systems with evolving memory characteristics. This provides a precise mathematical framework for classifying memory effects based on the relative weighting of different aspects of history.

\textbf{Inherent continuity with classical calculus.} 
Notably, the weighted estimate \eqref{eq:weighted_inequality} remains valid at $t=0$, where it reduces to the classical bound $|\dot{x}(0)| \le C_W \|x\|_{\mathcal{C}^{1}}$. This continuity at the initial instant underscores that the memory‑weighted framework does not introduce artificial singularities; rather, it extends the classical derivative in a manner that continuously incorporates historical information while preserving consistency with ordinary differentiation at the starting point—an essential property for physical models that should reduce to established local theories in appropriate limits.

Collectively, these results furnish a nuanced portrait of $\mathscr{V}_{\alpha,\beta}$ as an operator whose regularity—viewed through the lens of functional analysis—is delicately tuned by the relative magnitude of its two memory exponents. The explicit dependence on $\beta(t)-\alpha(t)$, captured both pointwise and through global extrema, offers a transparent link between the choice of memory laws and the resulting analytical behavior. This structured approach to memory modeling—where distinct aspects of history are weighted separately—may prove valuable in describing systems whose memory characteristics evolve in time or differ between various physical mechanisms, such as in materials with complex relaxation spectra or in transport processes with multiple time-scale dependencies.
\end{remark}


\appendix
\section{Technical Lemmas on Logarithmic and Power Function Inequalities} \label{app:log_power_inequalities}

This appendix presents several elementary but useful inequalities that provide control 
over logarithmic singularities and power functions. They will be employed in the analysis 
of integral kernels with variable exponents and in establishing uniform estimates.

\begin{lemma}[Logarithmic power-control inequality]\label{lem:log_power}
    For any $\delta > 0$, there exists a constant $C_\delta > 0$ such that for 
    all $x \in (0,1]$,
    \begin{equation}\label{eq:log_power}
        |\ln x| \leq C_\delta \, x^{-\delta}.
    \end{equation}
    An admissible choice is $C_\delta = \frac{1}{\delta e}$.
\end{lemma}

\begin{proof}
    Consider the auxiliary function $f(x) = x^{\delta} |\ln x|$ for $x \in (0,1)$.  
    Its derivative is
    \[
    f'(x) = \delta x^{\delta-1} |\ln x| + x^{\delta} \cdot \frac{d}{dx}|\ln x|.
    \]
    For $x \in (0,1)$ we have $\frac{d}{dx}|\ln x| = -1/x$; hence
    \[
    f'(x) = x^{\delta-1}\bigl(\delta |\ln x| - 1\bigr).
    \]
    Setting $f'(x)=0$ yields the unique critical point $x_0 = e^{-1/\delta}$.  
    Examining the sign of $f'(x)$:
    \begin{itemize}
        \item For $x \in (0, x_0)$: $\delta|\ln x| > 1$, so $f'(x) > 0$;
        \item For $x \in (x_0, 1)$: $\delta|\ln x| < 1$, so $f'(x) < 0$.
    \end{itemize}
    Consequently $f$ attains its global maximum at $x_0$.  The maximum value is
    \[
    f(x_0) = (e^{-1/\delta})^{\delta} \cdot \frac{1}{\delta} = \frac{1}{\delta e}.
    \]
    Therefore $x^{\delta} |\ln x| \leq \frac{1}{\delta e}$ for all $x \in (0,1)$.  
    Multiplying by $x^{-\delta}$ yields $|\ln x| \leq \frac{1}{\delta e} x^{-\delta}$ 
    for $x \in (0,1)$.  At $x=1$ the inequality holds trivially because 
    $|\ln 1| = 0 \leq \frac{1}{\delta e}$.  Hence \eqref{eq:log_power} is valid on 
    the whole interval $(0,1]$ with $C_\delta = \frac{1}{\delta e}$.
\end{proof}

\begin{lemma}[Logarithmic growth-control inequality]\label{lem:log_growth}
    For any $\alpha_0 > 0$, there exists a constant $C_{\alpha_0} > 0$ such that for 
    all $x \in [1,\infty)$,
    \begin{equation}\label{eq:log_growth}
        |\ln x| \leq C_{\alpha_0} \, x^{\alpha_0}.
    \end{equation}
    An admissible choice is $C_{\alpha_0} = \frac{1}{\alpha_0}$.
\end{lemma}

\begin{proof}
    Define $g_{\alpha_0}(x) = \ln x - \frac{x^{\alpha_0}}{\alpha_0}$ for $x \ge 1$.  
    Its derivative is
    \[
    g_{\alpha_0}'(x) = \frac{1}{x} - x^{\alpha_0-1} = \frac{1 - x^{\alpha_0}}{x}.
    \]
    For all $x \ge 1$, we have $x^{\alpha_0} \ge 1$, whence $g_{\alpha_0}'(x) \le 0$.
    Therefore $g_{\alpha_0}$ is decreasing on $[1,\infty)$.

    Consequently,
    \[
    g_{\alpha_0}(x) \leq g_{\alpha_0}(1) = 0 - \frac{1}{\alpha_0} = -\frac{1}{\alpha_0}
    \quad \text{for all } x \ge 1.
    \]
    Hence $\ln x \leq \frac{x^{\alpha_0}}{\alpha_0} - \frac{1}{\alpha_0} < \frac{x^{\alpha_0}}{\alpha_0}$.  
    Because $\ln x \ge 0$ for $x \ge 1$, we obtain $|\ln x| = \ln x \leq \frac{1}{\alpha_0} x^{\alpha_0}$, 
    which is precisely \eqref{eq:log_growth} with $C_{\alpha_0} = \frac{1}{\alpha_0}$.
\end{proof}

\begin{remark}[Unified logarithmic bound]\label{rem:unified_log_bound}
    Combining Lemmas~\ref{lem:log_power} and \ref{lem:log_growth} yields a bound valid 
    for all positive $x$:
    for any $\delta > 0$ and $\alpha > 0$,
    \[
    |\ln x| \leq C_{\delta,\alpha}\,
    \begin{cases}
        x^{-\delta}, & 0 < x \le 1,\\[4pt]
        x^{\alpha},  & x \ge 1,
    \end{cases}
    \]
    where one may take $C_{\delta,\alpha} = \max\bigl(\frac{1}{\delta e},\; \frac{1}{\alpha}\bigr)$.  
    This unified estimate proves convenient when handling integrals whose kernels involve 
    logarithms together with power‑law factors.
\end{remark}

\begin{lemma}[Linear bound for power functions with exponent at most one]\label{lem:power_linear_bound}
    For all $u \geq 0$ and all $\gamma \in (0,1]$, the following inequality holds:
    \[
    u^{\gamma} \leq 1 + u.
    \]
\end{lemma}

\begin{proof}
    We distinguish two possibilities according to the value of the exponent $\gamma$.

    \vspace{0.5em}
    \noindent
    \textbf{Case 1: $\gamma = 1$.} 
    In this situation the inequality reduces to $u \leq 1 + u$, which is evidently true 
    for every $u \geq 0$.

    \vspace{0.5em}
    \noindent
    \textbf{Case 2: $0 < \gamma < 1$.} 
    
    \textbf{Subcase 2.1: $u = 0$.} 
    The statement becomes $0 \leq 1$, which holds trivially.
    
    \textbf{Subcase 2.2: $u > 0$.} 
    Define the auxiliary function $g(u) = u^{\gamma} - u$ for $u > 0$. 
    Its derivative is
    \[
    g'(u) = \gamma u^{\gamma-1} - 1.
    \]
    Since $\gamma - 1 < 0$, the function $u \mapsto u^{\gamma-1}$ is strictly decreasing on $(0,\infty)$.
    Setting $g'(u) = 0$ yields the unique critical point
    \[
    u_0 = \gamma^{1/(1-\gamma)} > 0,
    \]
    which satisfies $\gamma u_0^{\gamma-1} = 1$, i.e., $u_0^{\gamma-1} = 1/\gamma$.
    
    \vspace{0.5em}
    \noindent
    Because $u^{\gamma-1}$ is strictly decreasing:
    \begin{itemize}
        \item For $0 < u < u_0$: $u^{\gamma-1} > u_0^{\gamma-1} = 1/\gamma$, whence $g'(u) > 0$;
        \item For $u > u_0$: $u^{\gamma-1} < u_0^{\gamma-1} = 1/\gamma$, whence $g'(u) < 0$.
    \end{itemize}
    Consequently $g$ is strictly increasing on $(0, u_0)$ and strictly decreasing on $(u_0, \infty)$;
    therefore $g$ attains its global maximum at $u_0$.

    \vspace{0.5em}
    \noindent
    Evaluating $g$ at the critical point:
    \begin{align*}
        g(u_0) &= u_0^{\gamma} - u_0 \\
               &= \gamma^{\gamma/(1-\gamma)} - \gamma^{1/(1-\gamma)} \\
               &= \gamma^{\gamma/(1-\gamma)} (1 - \gamma).
    \end{align*}
    Set $p = \dfrac{\gamma}{1-\gamma} > 0$. Then $g(u_0) = (1-\gamma) \gamma^{p}$.  
    Since $0 < \gamma < 1$ and $p > 0$, we have $0 < \gamma^{p} < 1$, and consequently
    \[
    g(u_0) < 1 - \gamma < 1.
    \]

    \vspace{0.5em}
    \noindent
    Thus, for every $u > 0$,
    \[
    g(u) \leq g(u_0) < 1.
    \]
    Rewriting this inequality gives $u^{\gamma} - u < 1$, i.e.,
    \[
    u^{\gamma} < 1 + u \quad \text{for all } u > 0.
    \]
    
    \vspace{0.5em}
    \noindent
    Combining Subcases 2.1 and 2.2, we obtain for all $u \geq 0$,
    \[
    u^{\gamma} \leq 1 + u \qquad \text{whenever } 0 < \gamma < 1.
    \]
    
    \vspace{0.5em}
    \noindent
    Finally, merging Cases 1 and 2 establishes the inequality for every $\gamma \in (0,1]$ 
    and every $u \geq 0$.
\end{proof}
\section{Auxiliary convergence results}\label{app:uniform_convergence}

This appendix provides a fundamental yet frequently employed result concerning the uniform convergence 
of products of function sequences. It is utilized in the proof of 
Theorem~\ref{thm:continuous_dependence} to analyze the product $A_n(t)I_n(t)$.

\begin{lemma}[Uniform convergence of products of uniformly convergent and bounded sequences]\label{lem:product_uniform_convergence}
    Let $E \subset \mathbb{R}$ and let $\{f_{1,n}\}, \{f_{2,n}\}, \dots, \{f_{k,n}\}$ 
    be $k$ sequences of real-valued functions defined on $E$.  Assume that:
    \begin{enumerate}
        \item For each $i = 1,\dots,k$, $f_{i,n}$ converges uniformly on $E$ 
              to a limit function $f_i$, i.e.,
              \[
              \lim_{n\to\infty} \sup_{x\in E} \bigl| f_{i,n}(x) - f_i(x) \bigr| = 0.
              \]
        \item Each sequence is uniformly bounded on $E$: there exist constants 
              $M_i > 0$ such that for all $n \geq 1$ and all $x \in E$,
              \[
              \bigl| f_{i,n}(x) \bigr| \leq M_i.
              \]
    \end{enumerate}
    Define the product sequences
    \[
    F_n(x) = \prod_{i=1}^{k} f_{i,n}(x), \qquad 
    F(x)   = \prod_{i=1}^{k} f_i(x), \qquad x \in E.
    \]
    Then $F_n$ converges uniformly on $E$ to $F$.
\end{lemma}

\begin{proof}
    The proof is organized into five systematic steps.
    
    \vspace{0.5em}
    \noindent
    \textbf{Step 1: Uniform boundedness of the limit functions.}
    From condition (2), for each fixed $x \in E$ and each index $i$,
    the real numbers $a_n = f_{i,n}(x)$ satisfy $|a_n| \leq M_i$ for all $n$.
    Since uniform convergence implies pointwise convergence, $f_{i,n}(x) \to f_i(x)$.
    We now establish that the limit $a = f_i(x)$ also satisfies $|a| \leq M_i$.
    
    Suppose, for a contradiction, that $|a| > M_i$.  Set $\upsilon = |a| - M_i > 0$.
    By pointwise convergence, there exists an integer $N$ such that for all $n > N$,
    $|a_n - a| < \upsilon/2$.  Applying the reverse triangle inequality,
    \[
    |a_n| = |a - (a - a_n)| \geq |a| - |a - a_n| > |a| - \frac{\upsilon}{2}.
    \]
    Substituting $|a| = M_i + \upsilon$ yields
    \[
    |a_n| > (M_i + \upsilon) - \frac{\upsilon}{2} = M_i + \frac{\upsilon}{2} > M_i,
    \]
    which contradicts $|a_n| \leq M_i$.  Hence $|a| \leq M_i$, i.e.,
    \[
    \bigl| f_i(x) \bigr| \leq M_i \quad \text{for all } x \in E.
    \]
    
    Set $M_0 = \max\{M_1,\dots,M_k\}$.  Then for every $i$, every $n$, and every $x \in E$,
    \[
    \bigl| f_{i,n}(x) \bigr| \leq M_0, \qquad 
    \bigl| f_i(x) \bigr| \leq M_0.
    \]
    If $M_0 = 0$, then all functions vanish identically and the conclusion holds trivially. 
    We therefore assume $M_0 > 0$ in the remainder of the proof.
    
    \vspace{0.5em}
    \noindent
    \textbf{Step 2: A telescoping identity for products.}
    For arbitrary real numbers $a_1,\dots,a_k, b_1,\dots,b_k$ we establish the identity
    \begin{equation}\label{eq:product_telescope}
    \prod_{i=1}^{k} a_i - \prod_{i=1}^{k} b_i
    = \sum_{j=1}^{k} \Biggl( \Bigl(\prod_{i=1}^{j-1} a_i\Bigr) 
                            (a_j - b_j) 
                            \Bigl(\prod_{i=j+1}^{k} b_i\Bigr) \Biggr),
    \end{equation}
    where we adopt the following convention for indexed products: for integers $p, q$ 
    (representing the lower and upper limits of the product index), we define
    \[
    \prod_{i=p}^{q} c_i := 
    \begin{cases}
    1, & \text{if } p > q \ (\text{the product contains no factors} \\
        & \text{and is defined to be } 1, \text{ the \textbf{empty product}}), \\[6pt]
    c_p \cdot c_{p+1} \cdots c_q, & \text{if } p \le q .
    \end{cases}
    \]
    In particular, when $p = q$ the product reduces to the single factor $c_p$.
    
    \textbf{Proof of the identity by induction on $k$.}
    
    \emph{Base case $k = 1$:}  
    Left‑hand side: $a_1 - b_1$.  
    Right‑hand side for $j = 1$:
    \[
    \prod_{i=1}^{0} a_i = 1 \ (\text{empty product}), \quad
    \prod_{i=2}^{1} b_i = 1 \ (\text{empty product}),
    \]
    so the right‑hand side equals $1 \cdot (a_1 - b_1) \cdot 1 = a_1 - b_1$.  
    Hence the identity holds for $k = 1$.
    
    \emph{Induction hypothesis:} Assume the identity is valid for some $k \geq 1$.
    
    \emph{Induction step $k \to k+1$:}
    \begin{align*}
        &\prod_{i=1}^{k+1} a_i - \prod_{i=1}^{k+1} b_i \\
        &= \Bigl( \prod_{i=1}^{k} a_i \Bigr) a_{k+1} 
           - \Bigl( \prod_{i=1}^{k} b_i \Bigr) b_{k+1} \\
        &= \Bigl( \prod_{i=1}^{k} a_i \Bigr) a_{k+1} 
           - \Bigl( \prod_{i=1}^{k} a_i \Bigr) b_{k+1}
           + \Bigl( \prod_{i=1}^{k} a_i \Bigr) b_{k+1}
           - \Bigl( \prod_{i=1}^{k} b_i \Bigr) b_{k+1} \\
        &= \Bigl( \prod_{i=1}^{k} a_i \Bigr) (a_{k+1} - b_{k+1})
           + b_{k+1} \Bigl( \prod_{i=1}^{k} a_i - \prod_{i=1}^{k} b_i \Bigr).
    \end{align*}
    Applying the induction hypothesis to the difference enclosed in parentheses,
    \[
    \prod_{i=1}^{k} a_i - \prod_{i=1}^{k} b_i
    = \sum_{j=1}^{k} \Biggl( \Bigl(\prod_{i=1}^{j-1} a_i\Bigr) 
                            (a_j - b_j) 
                            \Bigl(\prod_{i=j+1}^{k} b_i\Bigr) \Biggr),
    \]
    we obtain
    \begin{align*}
        &\prod_{i=1}^{k+1} a_i - \prod_{i=1}^{k+1} b_i \\
        &= \Bigl( \prod_{i=1}^{k} a_i \Bigr) (a_{k+1} - b_{k+1}) \\
        &\quad + b_{k+1} \sum_{j=1}^{k} 
           \Biggl( \Bigl(\prod_{i=1}^{j-1} a_i\Bigr) 
                  (a_j - b_j) 
                  \Bigl(\prod_{i=j+1}^{k} b_i\Bigr) \Biggr).
    \end{align*}
    In each term of the second sum, the factor $b_{k+1}$ can be absorbed into the 
    final product, yielding $\prod_{i=j+1}^{k+1} b_i$.  
    The first term corresponds to the index $j = k+1$, since in this case 
    $\prod_{i=1}^{j-1} a_i = \prod_{i=1}^{k} a_i$ (by the definition above, when $j = k+1$, 
    the product runs from $i=1$ to $i=k$) and 
    $\prod_{i=j+1}^{k+1} b_i$ is an empty product (as $j+1 = k+2 > k+1$) and equals $1$.  
    Consequently, the entire expression equals
    \[
    \sum_{j=1}^{k+1} \Biggl( \Bigl(\prod_{i=1}^{j-1} a_i\Bigr) 
                            (a_j - b_j) 
                            \Bigl(\prod_{i=j+1}^{k+1} b_i\Bigr) \Biggr),
    \]
    which is precisely \eqref{eq:product_telescope} with $k$ replaced by $k+1$.  
    By the principle of mathematical induction, the identity holds for every positive integer $k$.
    
    \vspace{0.5em}
    \noindent
    \textbf{Step 3: Applying the identity to $F_n(x) - F(x)$.}
    Take $a_i = f_{i,n}(x)$ and $b_i = f_i(x)$ in \eqref{eq:product_telescope}.
    Utilizing the uniform bounds established in Step 1,
    \[
    \Bigl| \prod_{i=1}^{j-1} f_{i,n}(x) \Bigr| \leq M_0^{\,j-1}, \qquad
    \Bigl| \prod_{i=j+1}^{k} f_i(x) \Bigr| \leq M_0^{\,k-j},
    \]
    for every index $j = 1,\dots,k$, where these bounds follow from the fact that each factor 
    is bounded in absolute value by $M_0$.  Consequently,
    \begin{align}
    |F_n(x) - F(x)|
    &\leq \sum_{j=1}^{k} 
         \Bigl| \prod_{i=1}^{j-1} f_{i,n}(x) \Bigr|
         \bigl| f_{j,n}(x) - f_j(x) \bigr|
         \Bigl| \prod_{i=j+1}^{k} f_i(x) \Bigr| \notag \\
    &\leq \sum_{j=1}^{k} M_0^{\,j-1} 
           \bigl| f_{j,n}(x) - f_j(x) \bigr|
           M_0^{\,k-j} \notag \\
    &= M_0^{\,k-1} \sum_{j=1}^{k} \bigl| f_{j,n}(x) - f_j(x) \bigr|. \label{eq:product_bound}
    \end{align}
    
    \vspace{0.5em}
    \noindent
    \textbf{Step 4: Uniform convergence of the individual factors.}
    Let $\zeta > 0$ be arbitrary.  Since each $f_{i,n}$ converges uniformly to $f_i$,
    there exists, for each $i$, an integer $N_i$ such that for all $n > N_i$,
    \[
    \sup_{x\in E} \bigl| f_{i,n}(x) - f_i(x) \bigr| < \frac{\zeta}{k M_0^{\,k-1}}.
    \]
    Set $N_0 = \max\{N_1,\dots,N_k\}$.
    
    \vspace{0.5em}
    \noindent
    \textbf{Step 5: Completion of the uniform convergence proof.}
    For every $n > N_0$ and every $x \in E$, inequality \eqref{eq:product_bound} combined with 
    the estimate from Step 4 yields
    \[
    |F_n(x) - F(x)| 
    < M_0^{\,k-1} \cdot k \cdot \frac{\zeta}{k M_0^{\,k-1}} = \zeta.
    \]
    Since the integer $N_0$ is independent of $x$, we have demonstrated that
    \[
    \sup_{x\in E} |F_n(x) - F(x)| < \zeta \qquad \text{for all } n > N_0,
    \]
    which is precisely the definition of uniform convergence of $F_n$ to $F$ on $E$.
\end{proof}

\begin{remark}
    This lemma represents a standard result in analysis; we present a complete proof here 
    for completeness and to maintain a self‑contained exposition.  The argument demonstrates 
    that the product of finitely many function sequences converges uniformly whenever each 
    factor sequence converges uniformly and is uniformly bounded.
\end{remark}

\section{Systematic Construction of Error-Control Functions} 
\label{app:error_control}

This appendix provides a self‑contained construction of auxiliary functions $\epsilon(s)$ that furnish uniform bounds for the Taylor remainder $r(h)=x(h)-x(0)-\dot{x}(0)h$ on intervals $[0,s]$. Unlike the pointwise limit $r(h)/h\to0$, the functions $\epsilon(s)$ capture the uniform decay rate, which is essential for handling integrals involving $r(\tau)$ over the whole interval $[0,t]$. The construction will be invoked in Step 6 of the proof of Theorem~\ref{thm:uniform_memory_limit}, and all required properties are established rigorously.

\subsection{Basic Setup and Notation}

Throughout this appendix we work on the fixed time interval $I = [0, T]$ with $T > 0$, exactly as in the main text. Let $x \in \mathcal{C}^{1}(I)$ be a continuously differentiable real-valued function, and let
\[
a_0 := \dot{x}(0)
\]
denote the classical derivative of $x$ at the initial instant $t = 0$.

\begin{definition}[Remainder function] \label{def:remainder_function_app}
For each $h \in I$, define the \textbf{remainder function}
\[
r(h) := x(h) - x(0) - a_0 h.
\]
Thus $r(h)$ is the remainder term in the first-order Taylor expansion of $x$ about the origin.
\end{definition}

To obtain a \emph{uniform} quantitative control over this decay—a control that will be essential for the limit analysis in the main proof—we introduce the following central object.

\begin{definition}[Error‑control function $\epsilon(s)$] \label{def:epsilon_function_app}
For $s > 0$, define
\[
\epsilon(s) := \sup_{0 < h \leq s} \Bigl| \frac{r(h)}{h} \Bigr|,
\]
and set $\epsilon(0) := 0$ by convention.
\end{definition}

Since our application concerns the right‑hand limit $t \to 0^{+}$, we restrict $h$ to be positive in the supremum. (For a bilateral analysis one would simply replace $0 < h \leq s$ by $0 < |h| \leq s$.)

\subsection{Well‑Posedness and Fundamental Properties}

We begin by establishing that the error‑control function $\epsilon(s)$ constitutes a well‑defined real‑valued function. A pivotal observation is that the supremum appearing in Definition~\ref{def:epsilon_function_app} is not merely finite but is actually attained; this follows from the continuity of a suitably extended version of the pointwise ratio $r(h)/h$.

\begin{lemma}[Continuity of the extended ratio function] \label{lem:continuity_aux_app}
Define the auxiliary function $\tilde{f}: [0, T] \to \mathbb{R}$ by
\[
\tilde{f}(h) := 
\begin{cases}
\displaystyle \frac{r(h)}{h}, & h \in (0, T], \\[10pt]
0, & h = 0.
\end{cases}
\]
Then $\tilde{f}$ is continuous on the compact interval $[0, T]$. (On the open‑closed interval $(0, T]$, $\tilde{f}$ coincides pointwise with $r(h)/h$.)
\end{lemma}

\begin{proof}
We verify continuity by examining the two components of the domain separately.

\textbf{Continuity on $(0, T]$.} For $h > 0$, the numerator $r(h) = x(h) - x(0) - a_0 h$ is continuous owing to the inclusion $x \in \mathcal{C}^{1}(I) \subset \mathcal{C}(I)$. The denominator $h$ is continuous and non‑zero throughout $(0, T]$. Consequently, $\tilde{f}$, being the quotient of two continuous functions with a non‑vanishing denominator, is continuous on $(0, T]$.

\textbf{Right‑continuity at the endpoint $h = 0$.} By the definition of the (right‑hand) derivative at the origin,
\[
\lim_{h \to 0^{+}} \tilde{f}(h) = \lim_{h \to 0^{+}} \frac{r(h)}{h}
= \lim_{h \to 0^{+}} \Bigl( \frac{x(h)-x(0)}{h} - a_0 \Bigr) = 0 = \tilde{f}(0).
\]
Thus $\tilde{f}$ is right‑continuous at $h = 0$.

Since $\tilde{f}$ is continuous on $(0, T]$ and right‑continuous at the left endpoint, it follows that $\tilde{f}$ is continuous on the entire interval $[0, T]$. In particular, its restriction to any closed subinterval $[0, s]$ with $0 < s \le T$ remains continuous.
\end{proof}

\begin{proposition}[Existence and finiteness of $\epsilon(s)$] \label{prop:epsilon_finite_app}
For every $s \in [0, T]$, the quantity $\epsilon(s)$ defined in Definition~\ref{def:epsilon_function_app} exists (the supremum is attained) and constitutes a finite non‑negative real number.
\end{proposition}

\begin{proof}
The case $s = 0$ is immediate, since $\epsilon(0) = 0$ by definition.

Now fix $s > 0$. Lemma~\ref{lem:continuity_aux_app} implies that the extended function $\tilde{f}$ is continuous on every compact subinterval of $[0, T]$; in particular, it is continuous on $[0, s]$. The composition $|\tilde{f}|$, defined by $|\tilde{f}|(h) = |\tilde{f}(h)|$, is therefore likewise continuous on $[0, s]$.

A key observation is that, because $\tilde{f}(0)=0$, the two suprema
\[
\sup_{0 < h \leq s} \bigl| \tilde{f}(h) \bigr| \quad\text{and}\quad \sup_{0 \leq h \leq s} \bigl| \tilde{f}(h) \bigr|
\]
coincide. Indeed, adjoining the point $h=0$ (where $|\tilde{f}(0)|=0$) cannot elevate the supremum. Hence
\[
\epsilon(s) = \sup_{0 < h \leq s} \bigl| \tilde{f}(h) \bigr| = \sup_{0 \leq h \leq s} \bigl| \tilde{f}(h) \bigr|.
\]

Since $|\tilde{f}|$ is continuous on the compact interval $[0, s]$, the extreme‑value theorem guarantees that it attains its maximum at some point $h_s \in [0, s]$. Therefore
\[
\epsilon(s) = \max_{0 \leq h \leq s} \bigl| \tilde{f}(h) \bigr| = \bigl| \tilde{f}(h_s) \bigr| < +\infty.
\]
Finally, $\epsilon(s) \geq 0$ because it is defined as a supremum of absolute values.
\end{proof}

\begin{remark}
The fact that the supremum in the definition of $\epsilon(s)$ is actually attained (i.e., $\epsilon(s)$ equals a maximum) follows from the continuity of $\tilde{f}$ (Lemma~\ref{lem:continuity_aux_app}) coupled with the compactness of $[0,s]$. For the proof of Theorem~\ref{thm:uniform_memory_limit}, two attributes of $\epsilon(s)$ are essential: 
\begin{enumerate}
    \item $\epsilon(s) < \infty$ for every $s>0$;
    \item $|r(h)/h| \leq \epsilon(s)$ for all $0 < h \leq s$ (uniform control).
\end{enumerate}
\end{remark}

\subsection{Limit Behavior at the Origin}

The utility of $\epsilon(s)$ stems from the fact that it furnishes a uniform upper bound for $|r(h)/h|$ on $[0,s]$, and---as demonstrated by the following proposition---this bound itself decays to zero as $s\to0^{+}$.

\begin{proposition}[Vanishing limit of $\epsilon(s)$] \label{prop:epsilon_limit_zero_app}
\[
\lim_{s \to 0^{+}} \epsilon(s) = 0.
\]
\end{proposition}

\begin{proof}
From Definition~\ref{def:epsilon_function_app}, $\epsilon(s) = \sup_{0 < h \leq s} |r(h)/h|$. By Lemma~\ref{lem:continuity_aux_app}, the function $\tilde{f}$ coincides with $r(h)/h$ for $h>0$ and satisfies $\tilde{f}(0)=0$. Moreover, the proof of Proposition~\ref{prop:epsilon_finite_app} shows that $\epsilon(s) = \sup_{0 \leq h \leq s} |\tilde{f}(h)|$ (since adjoining the point $h=0$, where $|\tilde{f}(0)|=0$, does not alter the supremum).

Let $\varepsilon_0 > 0$ be arbitrary. Since $\tilde{f}$ is continuous at $0$ and $\tilde{f}(0)=0$, there exists $\delta > 0$ such that for all $h$ satisfying $0 \leq h < \delta$, we have $|\tilde{f}(h)| < \varepsilon_0$.

Now select any $s$ with $0 < s < \delta$. For every $h \in [0, s]$ we have $h < \delta$, and consequently $|\tilde{f}(h)| < \varepsilon_0$. By the definition of the supremum,
\[
\epsilon(s) = \sup_{0 \leq h \leq s} |\tilde{f}(h)| \leq \varepsilon_0.
\]

Since $\varepsilon_0 > 0$ was arbitrary, this establishes $\lim_{s \to 0^{+}} \epsilon(s) = 0$.
\end{proof}

\subsection{Monotonicity and the Uniform Control Parameter $\epsilon_{t}$}

In the principal proof we shall require estimates that are uniform over an entire interval $[0,t]$, simultaneously governing $r(\tau)$ for all $0\leq\tau\leq t$ as well as $r(t)$ itself. The monotonicity of $\epsilon(s)$ and the introduction of a compact notation $\epsilon_t$ will facilitate this analysis.

\begin{lemma}[Monotonicity of $\epsilon(s)$] \label{lem:epsilon_monotone_app}
The error‑control function $\epsilon(s)$ is non‑decreasing on $[0, T]$; that is, for any $0 \leq s_{1} < s_{2} \leq T$,
\[
\epsilon(s_{1}) \leq \epsilon(s_{2}).
\]
\end{lemma}

\begin{proof}
Recall from Definition~\ref{def:epsilon_function_app} that $\epsilon(s) = \sup_{0<h\leq s} |\tilde{f}(h)|$, where $\tilde{f}$ denotes the continuous extension from Lemma~\ref{lem:continuity_aux_app}. For $0 \leq s_1 < s_2 \leq T$, the inclusion $\{h : 0<h\leq s_1\} \subset \{h : 0<h\leq s_2\}$ yields
\[
\epsilon(s_1) = \sup_{0<h\leq s_1} |\tilde{f}(h)| \leq \sup_{0<h\leq s_2} |\tilde{f}(h)| = \epsilon(s_2). \qedhere
\]
\end{proof}

Thus enlarging the interval $[0,s]$ cannot diminish the maximal relative error $|r(h)/h|$.

Since $\epsilon(s)$ is non‑decreasing, the supremum over $s\in[0,t]$ coincides with $\epsilon(t)$. This observation motivates the following compact notation, which will streamline the estimates in the main proof.

\begin{definition}[Uniform error‑control parameter $\epsilon_{t}$] \label{def:epsilon_t_app}
For $t \in [0, T]$, define
\[
\epsilon_{t} := \sup_{0 \leq s \leq t} \epsilon(s).
\]
\end{definition}

\begin{proposition}[Explicit form and properties of $\epsilon_{t}$] \label{prop:epsilon_t_explicit_app}
For every $t \in [0, T]$,
\[
\epsilon_{t} = \epsilon(t).
\]
Consequently, $\epsilon_{t}$ inherits the essential properties of $\epsilon(t)$:
\begin{enumerate}
    \item $\epsilon_{t} \geq 0$;
    \item If $0 \leq t_{1} < t_{2} \leq T$, then $\epsilon_{t_{1}} \leq \epsilon_{t_{2}}$ (monotonicity);
    \item $\displaystyle \lim_{t \to 0^{+}} \epsilon_{t} = 0$.
\end{enumerate}
\end{proposition}

\begin{proof}
By Lemma~\ref{lem:epsilon_monotone_app}, $\epsilon(s)$ is non‑decreasing. Therefore for a fixed $t > 0$,
\[
\sup_{0 \leq s \leq t} \epsilon(s) = \epsilon(t),
\]
since $\epsilon(s) \leq \epsilon(t)$ for all $s \in [0,t]$ and the value $\epsilon(t)$ is attained at $s=t$. The left‑hand side is precisely $\epsilon_t$ according to Definition~\ref{def:epsilon_t_app}, whence $\epsilon_t = \epsilon(t)$. The case $t = 0$ follows directly from the conventions $\epsilon_0 = \epsilon(0) = 0$.

Properties (1) and (2) are immediate consequences of the corresponding properties of $\epsilon(t)$. Property (3) follows from Proposition~\ref{prop:epsilon_limit_zero_app}:
\[
\lim_{t \to 0^{+}} \epsilon_{t} = \lim_{t \to 0^{+}} \epsilon(t) = 0. \qedhere
\]
\end{proof}

\subsection{Application: Uniform Bounds for the Remainder}

With the preceding construction we can now formulate a uniform estimate for the remainder function $r(h)$. This estimate will play an essential role in the proof of the main theorem.

\begin{lemma}[Uniform bound for the remainder] \label{lem:uniform_error_bound_app}
Assume $x \in \mathcal{C}^{1}(I)$ and let $r(h)$, $\epsilon(s)$, and $\epsilon_{t}$ be as defined above. Then for any $t > 0$ and any $\tau$ satisfying $0 \leq \tau \leq t$, the following inequalities hold:
\begin{align}
|r(\tau)| &\leq \epsilon_{t} \, \tau, \label{eq:bound_tau_app} \\
|r(t)|   &\leq \epsilon_{t} \, t.   \label{eq:bound_t_app}
\end{align}
\end{lemma}

\begin{proof}
Combining Definition~\ref{def:epsilon_t_app} with Proposition~\ref{prop:epsilon_t_explicit_app} yields
\[
\epsilon_{t} = \epsilon(t) = \sup_{0 < \tau \leq t} \Bigl| \frac{r(\tau)}{\tau} \Bigr|.
\]

\textbf{Estimate for $r(\tau)$.} If $\tau = 0$, then $r(0) = 0$ and the inequality $|0| \leq \epsilon_{t} \cdot 0$ holds trivially. If $0 < \tau \leq t$, then $\tau$ belongs to the set over which the supremum is taken; consequently
\[
\Bigl| \frac{r(\tau)}{\tau} \Bigr| \leq \epsilon_{t}.
\]
Multiplying by the positive number $\tau$ gives $|r(\tau)| \leq \epsilon_{t} \tau$.

\textbf{Estimate for $r(t)$.} Since $t > 0$, it is itself an element of $\{ \tau : 0 < \tau \leq t \}$. Therefore
\[
\Bigl| \frac{r(t)}{t} \Bigr| \leq \epsilon_{t},
\]
and multiplication by $t$ yields $|r(t)| \leq \epsilon_{t} t$.
\end{proof}

\addcontentsline{toc}{section}{REFERENCES}
\begin{thebibliography}{100}

\bibitem{bhat1992characterization}
B. V. Bhat, \textit{On a characterization of velocity maps in the space of observables}, Pac. J. Math. \textbf{152} (1992), 1--14.

\bibitem{kovavcik2013velocity}
S. Kov{\'a}{\v{c}}ik, P. Pre{\v{s}}najder, \textit{The velocity operator in quantum mechanics in noncommutative space}, J. Math. Phys. \textbf{54} (2013).

\bibitem{bae2015global}
H. Bae, R. Granero-Belinch{\'o}n, \textit{Global existence for some transport equations with nonlocal velocity}, Adv. Math. \textbf{269} (2015), 197--219.

\bibitem{maurel2024computing}
T. Maurel--Oujia, K. Matsuda, K. Schneider, \textit{Computing differential operators of the particle velocity in moving particle clouds using tessellations}, J. Comput. Phys. \textbf{498} (2024), 112658.

\bibitem{giannakis2021delay}
D. Giannakis, \textit{Delay-coordinate maps, coherence, and approximate spectra of evolution operators}, Res. Math. Sci. \textbf{8} (2021), no.~1, 8.

\bibitem{bagarello2013phenomenological}
F. Bagarello, F. Oliveri, \textit{A phenomenological operator description of interactions between populations with applications to migration}, Math. Models Methods Appl. Sci. \textbf{23} (2013), 471--492.

\bibitem{jasiczak2018coburn}
M. Jasiczak, \textit{Coburn--Simonenko theorem and invertibility of Toeplitz operators on the space of real analytic functions}, J. Oper. Theory \textbf{79} (2018), no.~2, 327--344.

\bibitem{albanese2025topologizability}
A. A. Albanese, H. Ariza, \textit{Topologizability and Related Properties of the Iterates of Composition Operators in Gelfand-Shilov Classes}, Integr. Equ. Oper. Theory \textbf{97} (2025), no.~3, 19.

\bibitem{vazquez2017classical}
J. L. V{\'a}zquez, A. de Pablo, F. Quir{\'o}s, A. Rodr{\'\i}guez, \textit{Classical solutions and higher regularity for nonlinear fractional diffusion equations}, J. Eur. Math. Soc. \textbf{19} (2017), 1949--1975.

\bibitem{chen2010heat}
Z.-Q. Chen, P. Kim, R. Song, \textit{Heat kernel estimates for the Dirichlet fractional Laplacian}, J. Eur. Math. Soc. \textbf{12} (2010), 1307--1329.

\bibitem{coimbra2003mechanics}
C. F. M. Coimbra, \textit{Mechanics with variable-order differential operators}, Ann. Phys. \textbf{515} (2003), 692--703.

\bibitem{umarov2009variable}
S. Umarov, S. Steinberg, \textit{Variable order differential equations with piecewise constant order-function and diffusion with changing modes}, Z. Anal. Anwend. \textbf{28} (2009), 431--450.

\bibitem{ramirez2010selection}
L. E. S. Ramirez, C. F. M. Coimbra, \textit{On the selection and meaning of variable order operators for dynamic modeling}, Int. J. Differ. Equ. \textbf{2010} (2010), 846107.

\bibitem{baliarsingh2022fractional}
P. Baliarsingh, L. Nayak, \textit{Fractional derivatives with variable memory}, Iran. J. Sci. Technol., Trans. A: Sci. \textbf{46} (2022), 849--857.

\bibitem{kachhia2021electromagnetic}
K. B. Kachhia, A. Atangana, \textit{Electromagnetic waves described by a fractional derivative of variable and constant order with non singular kernel}, Discrete Contin. Dyn. Ser. S \textbf{14} (2021).

\bibitem{kumar2019gegenbauer}
S. Kumar, P. Pandey, S. Das, \textit{Gegenbauer wavelet operational matrix method for solving variable-order non-linear reaction--diffusion and Galilei invariant advection--diffusion equations}, Comput. Appl. Math. \textbf{38} (2019), 162.

\bibitem{hasan2025novel}
M. M. A. Hasan, A. M. Alghanmi, S. M. AL-Mekhlafi, H. Al Ali, Z. Mukandavire, \textit{A Novel Crossover Dynamics of Variable-Order Fractal-Fractional Stochastic Diabetes Model: Numerical Simulations}, J. Math. \textbf{2025} (2025), 2986543.

\bibitem{kobelev2003statistical}
Ya. L. Kobelev, L. Ya. Kobelev, Yu. L. Klimontovich, \textit{Statistical Physics of Dynamic Systems with Variable Memory}, Dokl. Phys. \textbf{48} (2003).

\bibitem{iomin2025subordination}
A. Iomin, \textit{Subordination approach to Lyapunov exponents in random systems with memory}, Chaos \textbf{35} (2025).

\bibitem{gorska2023subordination}
K. G{\'o}rska, A. Horzela, \textit{Subordination and memory dependent kinetics in diffusion and relaxation phenomena}, Fract. Calc. Appl. Anal. \textbf{26} (2023), 480--512.

\bibitem{tarasov2024discrete}
V. E. Tarasov, \textit{Discrete maps with distributed memory fading parameter}, Comput. Appl. Math. \textbf{43} (2024), 113.

\bibitem{petkovic2011family}
M. S. Petkovi{\'c}, J. D{\v{z}}uni{\'c}, L. D. Petkovi{\'c}, \textit{A Family of Two-Point Methods with Memory for Solving Nonlinear Equations}, Appl. Anal. Discrete Math. (2011), 298--317.

\bibitem{ledesma2023differential}
C. T. Ledesma, J. A. Rodr{\'\i}guez, J. V. da C. Sousa, \textit{Differential equations with fractional derivatives with fixed memory length}, Rend. Circ. Mat. Palermo Ser. 2 \textbf{72} (2023), 635--653.

\bibitem{fanizza2024quantum}
M. Fanizza, J. Lumbreras, A. Winter, \textit{Quantum Theory in Finite Dimension Cannot Explain Every General Process with Finite Memory}, Commun. Math. Phys. \textbf{405} (2024), 50.

\bibitem{oultou2025history}
A. Oultou, H. Benaissa, \textit{A history-dependent hemivariational inequality for a non-stationary Navier-Stokes with long-memory}, SeMA J. \textbf{} (2025), 1--21.

\bibitem{rotstein2001phase}
H. G. Rotstein, S. Brandon, A. Novick-Cohen, A. Nepomnyashchy, \textit{Phase field equations with memory: the hyperbolic case}, SIAM J. Appl. Math. (2001), 264--282.

\bibitem{messaoudi2012general}
S. A. Messaoudi, A. Al-Shehri, \textit{General boundary stabilization of memory-type thermoelasticity with second sound}, Z. Anal. Anwend. \textbf{31} (2012), 441--461.

\bibitem{kerbal2024well}
S. Kerbal, N.-e. Tatar, N. Al-Salti, \textit{Well-posedness and stability of a fractional heat-conductor with fading memory}, Fract. Calc. Appl. Anal. \textbf{27} (2024), 1866--1905.

\bibitem{haddad1984functional}
G. Haddad, \textit{Functional viability theorems for differential inclusions with memory}, Ann. Inst. Henri Poincar{\'e} C, Anal. Non Lin{\'e}aire \textbf{1} (1984), 179--204.

\bibitem{li2021convergence}
J. Li, F. Xi, \textit{Convergence, boundedness, and ergodicity of regime-switching diffusion processes with infinite memory}, Front. Math. China \textbf{16} (2021), 499--523.

\bibitem{jumarie2009lagrangian}
G. Jumarie, \textit{From Lagrangian mechanics fractal in space to space fractal Schr{\"o}dinger’s equation via fractional Taylor’s series}, Chaos Soliton Fract \textbf{41} (2009), 1590--1604.

\bibitem{cortazar2024asymptotic}
C. Cort{\'a}zar, F. Quir{\'o}s, N. Wolanski, \textit{Asymptotic profiles for inhomogeneous heat equations with memory}, Math. Ann. \textbf{389} (2024), 3705--3746.

\bibitem{del2022fractional}
L. M. Del Pezzo, A. Quaas, J. D. Rossi, \textit{Fractional convexity}, Math. Ann. \textbf{383} (2022), 1687--1719.

\bibitem{kubica2024holder}
A. Kubica, K. Ryszewska, R. Zacher, \textit{H{\"o}lder continuity of weak solutions to evolution equations with distributed order fractional time derivative}, Math. Ann. \textbf{390} (2024), 2513--2592.

\bibitem{akdemir2023dependence}
A. O. Akdemir, H. D. Binh, D. O'Regan, A. T. Nguyen, \textit{The dependence on fractional orders of mild solutions to the fractional diffusion equation with memory}, Math. Methods Appl. Sci. \textbf{46} (2023), 1076--1095.

\bibitem{he2022solutions}
S. He, H. Wang, K. Sun, \textit{Solutions and memory effect of fractional-order chaotic system: A review}, Chin. Phys. B \textbf{31} (2022), 060501.

\bibitem{rodriguez2015generalized}
R. F. Rodr{\'\i}guez, J. Fujioka, \textit{Generalized hydrodynamic correlations and fractional memory functions.}, J. Non-Equilib. Thermodyn. \textbf{40} (2015).

\bibitem{edelman2014caputo}
M. Edelman, \textit{Caputo standard $\alpha$-family of maps: Fractional difference vs. fractional}, Chaos \textbf{24} (2014).

\bibitem{comlekoglu2017memory}
T. Comlekoglu, S. H. Weinberg, \textit{Memory in a fractional-order cardiomyocyte model alters properties of alternans and spontaneous activity}, Chaos \textbf{27} (2017).

\bibitem{ponce2021discrete}
R. Ponce, \textit{Discrete subdiffusion equations with memory}, Appl. Math. Optim. \textbf{84} (2021), 3475--3497.

\bibitem{ammari2021stabilization}
K. Ka{\"\i}s, F. Hassine, L. Robbiano, \textit{Stabilization of fractional evolution systems with memory}, J. Evol. Equ. \textbf{21} (2021), 831--844.

\bibitem{agram2019singular}
N. Agram, A. Bachouch, B. {\O}ksendal, F. Proske, \textit{Singular control optimal stopping of memory mean-field processes}, SIAM J. Math. Anal. \textbf{51} (2019), 450--468.

\bibitem{jiang2025study}
J. Jiang, B. Miao, \textit{A study of anomalous stochastic processes via generalizing fractional calculus}, Chaos \textbf{35} (2025).

\bibitem{georgiev2025sobolev}
V. Georgiev, M. Rastrelli, \textit{Sobolev spaces for singular perturbation of 2D Laplace operator}, Nonlinear Anal. \textbf{251} (2025), 113710.

\bibitem{junge2024entropic}
O. Junge, D. Matthes, B. Schmitzer, \textit{Entropic transfer operators}, Nonlinearity \textbf{37} (2024), 065004.

\bibitem{englivs2009toeplitz}
M. Engliš, \textit{Toeplitz operators and localization operators}, Trans. Amer. Math. Soc. \textbf{361} (2009), 1039--1052.

\bibitem{manoussakis2021small}
A. Manoussakis, A. Pelczar-Barwacz, \textit{Small operator ideals on the Schlumprecht and Schreier spaces}, J. Funct. Anal. \textbf{281} (2021), 109156.

\end{thebibliography}
\end{document}